\documentclass[11pt]{article}

\usepackage[T1]{fontenc}
\usepackage[english]{babel}
\usepackage[margin=1in]{geometry}
\usepackage{microtype}
\usepackage{cite}
\usepackage{amsmath,amssymb,amsfonts,mathtools}
\usepackage{amsthm}
\usepackage{graphicx}
\usepackage{algorithm}
\usepackage{algpseudocode}
\usepackage{float}
\usepackage{booktabs}
\usepackage{array}
\usepackage{multirow}
\usepackage{longtable}
\usepackage{adjustbox}
\usepackage{pdflscape}
\usepackage{enumitem}
\usepackage{xcolor}
\usepackage[colorlinks=true,allcolors=blue]{hyperref}
\usepackage[nameinlink,noabbrev]{cleveref}
\hypersetup{
  pdftitle={Resource-Aware Intrusion Detection in Infrastructure Networks: A Game-Theoretic Approach -- Technical Report},
  pdfsubject={Technical report},
  pdfkeywords={integrated sensing and communication, intrusion detection, graph security games, Nash equilibrium, Stackelberg equilibrium, resource allocation, edge computing}
}

\makeatletter
\renewcommand{\theHALG@line}{\thealgorithm.\arabic{ALG@line}}
\makeatother

\graphicspath{{figures/}}
\allowdisplaybreaks
\newtheorem{theorem}{Theorem}[section]

\newtheorem{proposition}{Proposition}[section]
\newtheorem{lemma}{Lemma}[section]
\theoremstyle{definition}

\newtheorem{example}{Example}[section]
\theoremstyle{remark}

\newcommand{\BR}{\operatorname{BR}}
\newcommand{\argmax}{\operatorname*{arg\,max}}

\newcommand{\pos}[1]{\left[#1\right]_+}
\newcommand{\supp}{\operatorname{supp}}
\newcommand{\R}{\mathbb{R}}

\title{
Resource-Aware Intrusion Detection in Infrastructure Networks:\\
A Game-Theoretic Approach\\[0.4em]
{\large Technical Report}
}

\author{
Xuanli Lin$^{*}$,
Zhaofeng Zhang$^{\dagger}$,
Zunzheng Zhang$^{*}$,
Kevin S. Chan$^{\ddagger}$,
Guoliang Xue$^{*}$\\[0.5em]
\small $^{*}$Arizona State University, Tempe, AZ,
\small \texttt{\{xlin54, zzhan621, xue\}@asu.edu}\\
\small $^{\dagger}$Delaware State University, Dover, DE,
\texttt{zzhang@desu.edu}\\
\small $^{\ddagger}$DEVCOM Army Research Laboratory, Adelphi, MD,
\texttt{kevin.s.chan.civ@army.mil}
}

\date{August 2026}

\begin{document}
\maketitle

\begin{abstract}
Infrastructure networks increasingly rely on distributed sensing
to detect intrusions before attackers reach valuable assets. Yet
sensing devices, communication resources, and edge server
capacity are limited, while intelligent attackers can adapt
their routes to the deployed defense. Motivated by integrated
sensing and communication (ISAC), we study how sensing and
processing resources should be allocated
under strategic interaction between a defender and an attacker.
We formulate their interaction
as a graph security game in which the defender deploys sensing
actions under resource and false alarm constraints, while the
attacker selects routes to valuable targets. We consider simultaneous play and settings in which the attacker observes either a
pure defender configuration or a mixed defender strategy.
Our analysis characterizes the existence, structure, and
computational complexity of the Nash and Stackelberg
equilibria, showing how the attacker's observation of the defense affects equilibrium
behavior and when optimal strategies become difficult to compute.
We develop algorithms that construct effective pure configurations and refine
restricted games for mixed Nash and mixed Stackelberg play. On enumerable instances, their
solutions have small mean normalized differences from fully enumerated references; the
methods also apply when exhaustive strategy enumeration is impractical. We also identify
conditions under which Nash and mixed Stackelberg payoffs are ordered or coincide.
\end{abstract}


\begingroup
\small
\tableofcontents
\endgroup
\newpage

\section{Introduction}
\label{sec:introduction}

Infrastructure networks increasingly use distributed sensing to detect intrusions before an
adversary reaches a valuable asset. Integrated sensing and communication (ISAC) makes such
protection possible on shared wireless infrastructure, but the available sensors,
communication resources, and edge server capacity remain limited
\cite{liu2022isac,devoti2020pasid}. A defender must therefore decide where to sense, which
capability to activate, and whether observations should be processed locally or at an edge
server. These coupled choices determine detection coverage, false alarms, deployment cost,
server load, and the sensing--communication tradeoff
\cite{an2023pdtradeoff,he2025versabeam}. Their effectiveness also depends on a strategic
attacker whose target and route choices respond to the information available about the
defense \cite{hu2025game}.

The central question is consequently not only which sensing configuration performs well,
but which pure configuration or mixed defender strategy remains effective under the information
available to the attacker. We use one graph model and one pair of pure action spaces to
study three game settings. In the simultaneous Nash game, neither player observes the
other player's realized action before choosing its own. In the pure Stackelberg setting,
the attacker observes the defender's pure sensing configuration before routing. In the mixed
Stackelberg setting, the attacker observes the distribution over complete
configurations but not the configuration drawn for the current interaction. This
progression isolates how the attacker's observation of the defense changes equilibrium
behavior, computational structure, and the value of hiding the realized configuration.

Security games allocate limited defensive resources across targets or constrained schedules
\cite{kiekintveld2009massive,korzhyk2011multiple,sinha2018stackelberg}. Graph security games
and network interdiction model route selection and exponentially large action spaces
\cite{basilico2012patrolling,jain2011double,letchford2013marginal,
vcerny2024layered,mai2024tackling,washburn1995interdiction,dahan2022network,
smith2020survey}, whereas adversarial network design modifies or protects the topology
itself \cite{ciftcioglu2017topology,pal2016editing,pal2017decentralized}. Intrusion
detection games consider adaptive intruders and false alarm constraints
\cite{bhargav2025sensor,hu2025game}, while integrated sensing, communication, and computing
work optimizes sensing, offloading, and edge computation
\cite{liu2023offloading,liu2025latency,wang2025collaborative,liu2026integrated,
wen2025survey,luo2025unified}. We connect these directions in a graph model with fixed topology that jointly selects monitored edges, sensing types, and processing modes under deployment
cost, edge server capacity, and an aggregate false alarm constraint. The model supports
general-sum payoffs that vary by target, multiplicative route detection, and attacks against
multiple targets.

The formulation exposes a sharp computational divide. For a pure defender configuration,
a logarithmic transformation turns the multiplicative probability of reaching a target
undetected into an additive shortest path objective. For a hidden realization drawn from a
mixed defender strategy, expected non-detection probability is instead a weighted sum of
path products. Correlation induced by the common hidden configuration prevents that
expectation from being represented by one set of exact additive edge weights in general,
as occurs in reliable path models with correlated failures
\cite{chang2007reachability}. This divide motivates a restricted game method related to
branch-and-price and double-oracle approaches \cite{jain2010arbitrary,jain2011double},
although the exact best response subproblems in the mixed games are themselves NP-hard.

The main contributions of this report are the following.
\begin{enumerate}[leftmargin=2.2em,itemsep=0.25em]
    \item We formulate the common graph game and make explicit how sensing type,
    processing mode, resource limits, false alarms, route detection, and payoffs assigned to
    each target determine the two players' expected utilities under pure and mixed strategies. We show that a pure Nash
    equilibrium need not exist, while a mixed Nash equilibrium always exists, and we give
    exact support conditions for an enumerated game in normal form.

    \item For an observed pure defender configuration, we derive an exact
    shortest path attacker response and prove its runtime. We show that an optimal
    pure Stackelberg defense exists under both strong and pessimistic tie-breaking, but
    optimizing that defense is NP-hard even with one perfect sensing type and only a
    cardinality budget. We then develop \texttt{RADAR} (Resource-Aware Downgrading with
    Adaptive Refinement), including its repair score, neighborhood construction,
    termination property, and runtime bound.

    \item For a mixed defender strategy whose realization remains hidden, we prove the
    existence of a mixed strong Stackelberg equilibrium and give a complete example in
    which the pessimistic problem has a supremum but no maximizing strategy. We establish
    NP-hardness of the exact attacker oracle, the pure defender oracle, and mixed
    Stackelberg optimization through explicit reductions. These results motivate
    \texttt{RGR}, a restricted game method guided by \texttt{RADAR} for mixed Nash and mixed
    Stackelberg computation. We derive the Stackelberg column score from the restricted
    linear program and state precisely what its generated action diagnostics do and do not
    certify.

    \item We characterize the relationship between the three games. We give the exact
    condition under which a pure Stackelberg outcome is also a pure Nash equilibrium,
    prove a minimax characterization for attacks on a single route, extend it to coordinated
    attacks on multiple routes under a common payoff ratio, bound the Nash--Stackelberg payoff
    differences, and identify the weighted zero-sum case in which mixed Nash and mixed
    Stackelberg payoffs coincide.

    \item We expand the numerical evaluation beyond the conference presentation.
    On enumerable instances, \texttt{RADAR} has a $3.64\%$ mean normalized gap from exact
    optimization over pure configurations, while \texttt{RGR}'s mean absolute normalized
    defender payoff differences from selected fully enumerated solutions are $0.466\%$ for
    Nash and
    $0.433\%$ for mixed Stackelberg play. The expanded results separate candidate modes
    and candidate counts, report quality by topology and counts of objective evaluations,
    include larger instances for which exhaustive comparison is unavailable, and check
    the predicted payoff identities and inequalities numerically.
\end{enumerate}

The remainder of the report is organized as follows. Section~\ref{sec:model} defines the common
model. Sections~\ref{sec:nash}--\ref{sec:mixed_stackelberg} analyze the simultaneous,
pure Stackelberg, and mixed Stackelberg games. Section~\ref{sec:relations}
compares their outcomes, and Section~\ref{sec:evaluation} presents the expanded evaluation.
The appendices contain the longer reductions, complete supporting proofs, and detailed
evaluation tables.

\section{System Model and Player Strategies}
\label{sec:model}

This section defines the network, the defender's sensing configurations, the attacker's route
actions, and the payoffs assigned to each target. These elements are shared by all three
games; only the order of play and the attacker's observations change.

\subsection{Network Environment}
\label{subsec:network}

Let $G=(\mathcal V,\mathcal E)$ be an undirected graph representing a monitored physical
environment. Vertices may represent entrances, junctions, checkpoints, or target areas,
while edges may represent corridors or links along which an intrusion can progress. The
sets $\mathcal S\subseteq\mathcal V$ and $\mathcal T\subseteq\mathcal V$ contain possible
attack sources and valuable targets, respectively, with
$\mathcal S\cap\mathcal T=\emptyset$. We assume that every target is reachable from at
least one source. The topology is fixed: the defender selects sensing and processing
actions on the existing edges but does not add, remove, or reroute network links.

\subsection{Defender Strategies}
\label{subsec:defender-strategies}

The defender can deploy sensing capabilities along edges $e\in\mathcal E$. Let
$\mathcal K=\{1,\ldots,K\}$ denote the available sensing types, each representing a
distinct sensing capability. Let $\mathcal M=\{0,1\}$ denote the processing modes, where
$m=0$ means local processing and $m=1$ means offloading to an edge server. A sensing
action
\[
 h=(k,e,m)
\]
deploys sensing type $k$ on edge $e$ using processing mode $m$. The set
$\mathcal H\subseteq\mathcal K\times\mathcal E\times\mathcal M$ contains all
operationally supported actions; unsupported combinations of sensing type, edge, and
processing mode are simply absent from $\mathcal H$.

For $h=(k,e,m)$, let $\rho_{kem}\in[0,1]$ be its effective detection probability,
$p^{\mathrm{fa}}_{kem}\in[0,1]$ its false alarm probability, $c_{ke}\ge0$ its deployment
cost, and $u_{ke}\ge0$ its server load when offloaded. The detection and false alarm
parameters may depend on processing mode because local and offloaded processing can use
different algorithms or data resolutions.

For each edge, define
\[
 \mathcal H_e=\{(k,e,m)\in\mathcal H\}.
\]
A defender configuration is a set $z\subseteq\mathcal H$ with at most one sensing action
per edge. Writing $z_e=z\cap\mathcal H_e$, the set of all configurations is
\[
 \mathcal Z=\{z\subseteq\mathcal H:|z_e|\le1,\ \forall e\in\mathcal E\}.
\]
The probability that an intrusion is detected while traversing edge $e$ under $z$ is
\begin{equation}
 p_e(z)=\sum_{(k,e,m)\in z_e}\rho_{kem}.
 \label{eq:edge_detection}
\end{equation}
The sum contains at most one nonzero term because $|z_e|\le1$.

Assuming that false alarm events are independent across deployed sensing actions
\cite{stolkin2009pdplacement,viswanathan1997distributed}, the probability that at least one
deployed action produces a false alarm is
\begin{equation}
 P_{\mathrm{fa}}(z)
 =1-\prod_{e\in\mathcal E}
 \left(1-\sum_{(k,e,m)\in z_e}p^{\mathrm{fa}}_{kem}\right).
 \label{eq:aggregate_false_alarm}
\end{equation}
Let $B_D>0$ be the deployment budget, $C_{\mathrm{srv}}>0$ the edge server capacity, and
$\bar P_{\mathrm{fa}}\in(0,1]$ the aggregate false alarm limit. Define the deployment cost
and server load of $z$ by
\begin{equation}
 b(z)=\sum_{(k,e,m)\in z}c_{ke},
 \qquad
 \ell(z)=\sum_{(k,e,m)\in z}m u_{ke}.
 \label{eq:resource_usage}
\end{equation}
The feasible defender action space is
\begin{equation}
\begin{aligned}
 \mathcal X=\{z\in\mathcal Z:
 &P_{\mathrm{fa}}(z)\le\bar P_{\mathrm{fa}},\quad
 b(z)\le B_D,\quad
 \ell(z)\le C_{\mathrm{srv}}\}.
\end{aligned}
\label{eq:defender_feasible_set}
\end{equation}
Each $z\in\mathcal X$ is a feasible pure defender strategy. The empty configuration has
zero cost, zero server load, and zero false alarm probability, so
$\emptyset\in\mathcal X$ and $\mathcal X$ is always nonempty under the stated parameter
ranges.

\subsection{Attacker Strategies}
\label{subsec:attacker-strategies}

An attack route $r$ is a simple path from a source $s\in\mathcal S$ to a target
$t\in\mathcal T$. Let $\mathcal E(r)$ be its edge set, $t(r)$ its target, and
$\mathcal R$ the finite set of all such routes. For a target $t$, write
\[
 \mathcal R_t=\{r\in\mathcal R:t(r)=t\}.
\]

For $z\in\mathcal Z$, assume that detection events on different edges are independent
conditional on the realized configuration $z$ \cite{chang2007reachability}. The route
non-detection probability, meaning the probability that the attacker reaches the target
undetected via $r$, is
\begin{equation}
 q_r(z)=\prod_{e\in\mathcal E(r)}\bigl(1-p_e(z)\bigr).
 \label{eq:route_survival}
\end{equation}
Conditional independence is imposed only after $z$ is fixed. When the defender later
randomizes over complete configurations, the common hidden realization can correlate the
unconditional sensing outcomes on different edges.

Restricting routes to simple paths loses no attacker optimum. Consider any source--target
walk containing a cycle. Removing the cycle preserves its source and target while removing
factors from the product in \eqref{eq:route_survival}. Every removed factor lies in $[0,1]$,
so the resulting route has non-detection probability at least as large. Because the attacker
utility defined below is strictly increasing in this probability, a cyclic walk cannot be
strictly better than the corresponding simple path.

The attacker may launch attacks against several targets simultaneously
\cite{korzhyk2011multiple}. A pure attacker strategy is a set of routes containing at most
one route to each target. Formally,
\begin{equation}
 \mathcal A=
 \left\{a\subseteq\mathcal R:
 \left|a\cap\mathcal R_t\right|\le1,\ \forall t\in\mathcal T\right\}.
 \label{eq:attacker_action_set}
\end{equation}
The empty action $\emptyset\in\mathcal A$ gives both players zero utility and represents
the attacker abstaining from attack. The model places no exogenous upper bound on the
number of simultaneously attacked targets; instead, the sign of the route utility
determines whether attacking each target is worthwhile.

Figure~\ref{fig:model} illustrates one pure strategy pair.
\begin{figure}[htbp]
    \centering
    \includegraphics[width=0.3\textwidth]{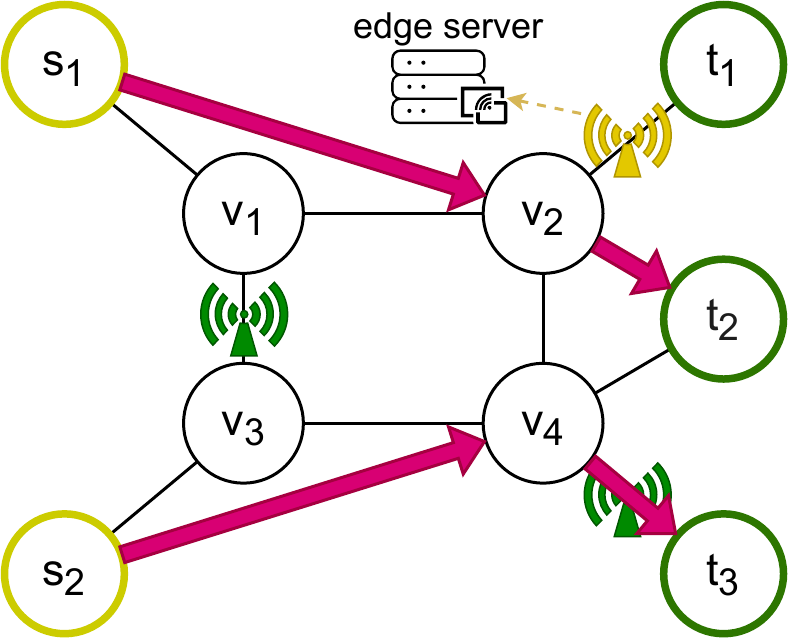}
    \caption{Example pure strategy pair. The defender deploys two locally processed
    sensing actions and one offloaded action; the attacker selects two routes to distinct
    targets.}
    \label{fig:model}
\end{figure}

\subsection{Utility Functions}
\label{subsec:utilities}

For target $t$, let $U_D^u(t)<0$ and $U_D^d(t)>0$ be the defender utilities when an
attack is undetected and detected, respectively. Define $U_A^u(t)>0$ and $U_A^d(t)<0$
analogously for the attacker. These sign assumptions make an undetected intrusion harmful
to the defender and beneficial to the attacker, while detection reverses those effects.
The magnitudes may vary by target, allowing the players to value the same target
differently.

For a pure defender configuration $z$ and route $r$, the defender's expected route utility
is
\begin{equation}
\begin{aligned}
 U_D(z,r)
 &=q_r(z)U_D^u(t(r))
   +\bigl(1-q_r(z)\bigr)U_D^d(t(r))\\
 &=U_D^d(t(r))-\Delta_D(t(r))q_r(z),
\end{aligned}
\label{eq:defender_route_utility}
\end{equation}
where
\[
 \Delta_D(t)=U_D^d(t)-U_D^u(t)>0.
\]
The attacker's expected route utility is
\begin{equation}
\begin{aligned}
 U_A(z,r)
 &=q_r(z)U_A^u(t(r))
   +\bigl(1-q_r(z)\bigr)U_A^d(t(r))\\
 &=U_A^d(t(r))+\Delta_A(t(r))q_r(z),
\end{aligned}
\label{eq:attacker_route_utility}
\end{equation}
where
\[
 \Delta_A(t)=U_A^u(t)-U_A^d(t)>0.
\]
Thus, for a fixed target, the defender route utility decreases with the non-detection
probability, whereas the attacker route utility increases with it.

Utilities are additive across the routes in an attacker action:
\begin{equation}
 U_D(z,a)=\sum_{r\in a}U_D(z,r),
 \qquad
 U_A(z,a)=\sum_{r\in a}U_A(z,r).
 \label{eq:action_utilities}
\end{equation}
A target omitted from $a$ contributes zero to both players. Additivity and the restriction
of at most one route per target imply that, once the best route to each target is known, the
attacker can decide independently whether to include that target. This separability is used
by the shortest path response in Section~\ref{sec:pure_stackelberg} and the candidate
generator in Section~\ref{sec:mixed_stackelberg}.

\begin{table}[htbp]
\centering
\caption{Frequently used notation.}
\label{tab:notation}
\begin{tabular}{@{}p{0.20\textwidth}p{0.72\textwidth}@{}}
\toprule
Symbol & Meaning \\
\midrule
$G=(\mathcal V,\mathcal E)$ & Monitored undirected graph \\
$\mathcal S,\mathcal T$ & Attack sources and valuable targets \\
$\mathcal H,\mathcal H_e$ & Supported sensing actions and actions supported on edge $e$ \\
$\mathcal Z,\mathcal X$ & All defender configurations and feasible configurations \\
$p_e(z),P_{\mathrm{fa}}(z)$ & Edge detection probability and aggregate false alarm probability \\
$b(z),\ell(z)$ & Deployment cost and edge server load \\
$\mathcal R,\mathcal R_t$ & All attack routes and routes terminating at target $t$ \\
$\mathcal A$ & Pure attacker actions, with at most one route per target \\
$q_r(z)$ & Route non-detection probability \\
$U_i^u(t),U_i^d(t)$ & Player $i\in\{D,A\}$ payoff at target $t$ under an undetected or detected attack \\
$\Delta_D(t),\Delta_A(t)$ & Payoff differences between detected and undetected outcomes \\
\bottomrule
\end{tabular}
\end{table}

The pure action spaces $\mathcal X$ and $\mathcal A$ are common to all formulations below.
Decisions are simultaneous in the Nash game. In the Stackelberg games, the defender acts
first and the attacker responds after observing either the realized pure configuration or
the mixed defender strategy but not its realization.

\section{Simultaneous Nash Game}
\label{sec:nash}

We begin with the simultaneous Nash game. The defender chooses $z\in\mathcal X$ and the
attacker chooses $a\in\mathcal A$ without observing the other player's realized action. A
Nash equilibrium makes these choices mutually stable: each player's strategy is a best
response to the other player's strategy.

\subsection{Pure Nash Equilibrium Need Not Exist}
\label{subsec:pure-ne}

A pure Nash equilibrium (NE) is a pair $(z^{\mathrm{NE}},a^{\mathrm{NE}})$ satisfying
\begin{align}
 U_D(z^{\mathrm{NE}},a^{\mathrm{NE}})
 &\ge U_D(z,a^{\mathrm{NE}}),
 &&\forall z\in\mathcal X,
 \label{eq:pure-ne-defender}\\
 U_A(z^{\mathrm{NE}},a^{\mathrm{NE}})
 &\ge U_A(z^{\mathrm{NE}},a),
 &&\forall a\in\mathcal A.
 \label{eq:pure-ne-attacker}
\end{align}
Route adaptation can prevent such stability even in a minimal instance.

\begin{theorem}
\label{thm:pure-ne-existence}
A pure Nash equilibrium need not exist.
\end{theorem}
\begin{proof}
Consider two edge-disjoint routes $r_1$ and $r_2$ leading to two different targets. The
defender can deploy at most one perfect sensor and can therefore cover at most one route.
Let $z_0$ denote the empty configuration, and let $z_j$ place a perfect sensor on $r_j$ for
$j\in\{1,2\}$. Thus,
\[
 \mathcal X=\{z_0,z_1,z_2\}.
\]
For both targets, set
\[
 U_A^u=3,\qquad U_A^d=-1,\qquad
 U_D^u=-1,\qquad U_D^d=1.
\]
An uncovered route remains undetected with probability one and gives payoff $(-1,3)$,
ordered as $(U_D,U_A)$. A perfectly covered route remains undetected with probability zero
and gives $(1,-1)$.
By additivity across targets, the nonempty part of the payoff matrix is
\begin{equation*}
\begin{array}{c|ccc}
 & \{r_1\} & \{r_2\} & \{r_1,r_2\}\\ \hline
 z_0 & (-1,3) & (-1,3) & (-2,6)\\
 z_1 & (1,-1) & (-1,3) & (0,2)\\
 z_2 & (-1,3) & (1,-1) & (0,2).
\end{array}
\end{equation*}
The empty attacker action gives $(0,0)$ against every defender action and is never an
attacker best response in this instance.

We now examine every possible defender action. Against $z_0$, both routes have positive
attacker utility, so the unique attacker best response is $\{r_1,r_2\}$. Holding that
action fixed, the defender improves from $-2$ to $0$ by changing to either $z_1$ or $z_2$.
Consequently, no pair with defender action $z_0$ satisfies
\eqref{eq:pure-ne-defender}--\eqref{eq:pure-ne-attacker}.

Against $z_1$, route $r_1$ gives attacker utility $-1$ and route $r_2$ gives utility $3$.
The unique attacker best response is therefore $\{r_2\}$. Holding $\{r_2\}$ fixed, the
defender improves from $-1$ under $z_1$ to $1$ under $z_2$. Hence no pair with defender
action $z_1$ is a pure NE\@. By symmetry, the unique attacker best response to $z_2$ is
$\{r_1\}$, and the defender improves by changing to $z_1$. These three cases exhaust
$\mathcal X$, so no pure Nash equilibrium exists.
\end{proof}

The failure is not caused by imperfect sensing, processing constraints, or complicated
routes. The attacker's best response changes with sensor placement, while each Nash
deviation holds the other player's current action fixed. The resulting best response cycle
prevents a stable pure pair.

\subsection{Mixed Nash Equilibrium}
\label{subsec:mixed-ne}

The players may instead randomize over complete pure actions. Let
$\Delta(\mathcal X)$ and $\Delta(\mathcal A)$ be the corresponding probability simplices,
and let $\sigma\in\Delta(\mathcal X)$ and $\mu\in\Delta(\mathcal A)$ denote defender and
attacker mixed strategies. A draw from $\sigma$ realizes one complete feasible
configuration; it does not fractionally deploy individual sensing actions.

The expected utilities are
\begin{align}
 \bar U_D(\sigma,\mu)
 &=\sum_{z\in\mathcal X}\sum_{a\in\mathcal A}
 \sigma(z)\mu(a)U_D(z,a),
 \label{eq:mixed-defender-utility}\\
 \bar U_A(\sigma,\mu)
 &=\sum_{z\in\mathcal X}\sum_{a\in\mathcal A}
 \sigma(z)\mu(a)U_A(z,a).
 \label{eq:mixed-attacker-utility}
\end{align}
We use the same notation when one argument is pure. For example,
$\bar U_D(\sigma,a)=\sum_z\sigma(z)U_D(z,a)$ and
$\bar U_D(z,\mu)=\sum_a\mu(a)U_D(z,a)$.

A mixed NE $(\sigma^{\mathrm{NE}},\mu^{\mathrm{NE}})$ satisfies
\begin{align}
 \bar U_D(\sigma^{\mathrm{NE}},\mu^{\mathrm{NE}})
 &\ge \bar U_D(z,\mu^{\mathrm{NE}}),
 &&\forall z\in\mathcal X,
 \label{eq:mixed-ne-defender}\\
 \bar U_A(\sigma^{\mathrm{NE}},\mu^{\mathrm{NE}})
 &\ge \bar U_A(\sigma^{\mathrm{NE}},a),
 &&\forall a\in\mathcal A.
 \label{eq:mixed-ne-attacker}
\end{align}
It is sufficient to check pure deviations. Indeed, the payoff from any mixed deviation is a
convex combination of the payoffs from its pure actions and therefore cannot exceed the
largest payoff from a pure deviation.

\begin{theorem}
\label{thm:mixed-ne-existence}
The simultaneous graph game has at least one mixed Nash equilibrium.
\end{theorem}
\begin{proof}
The defender action set $\mathcal X$ is finite because it is a subset of the power set of
the finite set $\mathcal H$ of sensing actions. The route set $\mathcal R$ is finite because a
finite graph has finitely many simple paths, and $\mathcal A$ is finite because it is a
subset of the power set of $\mathcal R$. The game is therefore a finite two-player game in
normal form. Nash's existence theorem for finite games guarantees a mixed equilibrium
\cite{nash1951noncooperative}.
\end{proof}

For the instance used in the proof of Theorem~\ref{thm:pure-ne-existence}, one mixed NE is
obtained by assigning probability $1/2$ to each of $z_1$ and $z_2$ and letting the attacker
choose $\{r_1,r_2\}$. Each route then gives the attacker expected utility $1$, so attacking
both gives utility $2$ and is optimal. Against the two-route action, both $z_1$ and $z_2$
give the defender utility zero, whereas $z_0$ gives $-2$. Thus both players are best
responding even though no pure action pair is stable.

\subsection{Support Characterization for Enumerated Games}
\label{subsec:ne-support}

When all pure actions are explicitly enumerated, candidate supports provide an exact way to
compute and verify mixed equilibria. Let nonempty sets $S_D\subseteq\mathcal X$ and
$S_A\subseteq\mathcal A$ be candidate supports. Set all probabilities outside these sets
to zero, and introduce payoff variables $v_D$ and $v_A$. Consider
\begin{align}
 \sum_{a\in S_A}\mu(a)U_D(z,a)&=v_D,
 &&z\in S_D,
 \label{eq:ne-support-defender-eq}\\
 \sum_{a\in S_A}\mu(a)U_D(z,a)&\le v_D,
 &&z\in\mathcal X\setminus S_D,
 \label{eq:ne-support-defender-ineq}\\
 \sum_{z\in S_D}\sigma(z)U_A(z,a)&=v_A,
 &&a\in S_A,
 \label{eq:ne-support-attacker-eq}\\
 \sum_{z\in S_D}\sigma(z)U_A(z,a)&\le v_A,
 &&a\in\mathcal A\setminus S_A,
 \label{eq:ne-support-attacker-ineq}
\end{align}
together with
\begin{equation}
 \sigma(z)\ge0,\quad \mu(a)\ge0,\qquad
 \sum_{z\in S_D}\sigma(z)=1,\quad
 \sum_{a\in S_A}\mu(a)=1.
 \label{eq:ne-support-probabilities}
\end{equation}

\begin{proposition}
\label{prop:ne-support-characterization}
Every feasible solution of
\eqref{eq:ne-support-defender-eq}--\eqref{eq:ne-support-probabilities} defines a mixed NE
whose actual supports are contained in $(S_D,S_A)$. Conversely, every mixed NE satisfies
this system when $S_D$ and $S_A$ are its exact supports.
\end{proposition}
\begin{proof}
Suppose the system is feasible. Every defender action assigned positive probability lies in
$S_D$ and gives expected payoff $v_D$ against $\mu$ by
\eqref{eq:ne-support-defender-eq}. Hence the defender's expected payoff under $\sigma$ is
also $v_D$. Every action outside $S_D$ gives payoff at most $v_D$ by
\eqref{eq:ne-support-defender-ineq}; actions inside $S_D$ give exactly $v_D$. Therefore no
pure defender deviation is profitable, and, by linearity, no mixed defender deviation is
profitable. The same argument using
\eqref{eq:ne-support-attacker-eq}--\eqref{eq:ne-support-attacker-ineq} shows that the
attacker's expected payoff is $v_A$ and that no attacker deviation is profitable. Thus
$(\sigma,\mu)$ is a mixed NE\@.

Conversely, let $(\sigma,\mu)$ be a mixed NE and let $S_D=\supp(\sigma)$ and
$S_A=\supp(\mu)$. Every action in the support of a best response must itself be a best
response. If, for example, some $z\in S_D$ gave strictly less than the maximum defender
payoff against $\mu$, then the weighted average under $\sigma$ would be strictly less than
that maximum, contradicting \eqref{eq:mixed-ne-defender}. Therefore every $z\in S_D$ gives
the common equilibrium payoff $v_D$, and every action outside $S_D$ gives at most $v_D$.
The attacker conditions follow identically. The support probabilities are nonnegative and
sum to one, so the equilibrium satisfies the complete system.
\end{proof}

Exhaustive enumeration of candidate support pairs is exact for an explicitly listed game,
but it scales poorly: the graph may encode exponentially many pure configurations and
routes, and the number of candidate support pairs is itself exponential in the size of the
game in normal form. The next two sections exploit the graph structure where possible and use
restricted action sets where exact enumeration is impractical.

\section{Pure Stackelberg Game}
\label{sec:pure_stackelberg}

We next consider a sequential setting in which reconnaissance reveals the defender's pure
sensing configuration. The defender acts first and selects $z\in\mathcal X$. The attacker observes
that complete configuration before choosing $a\in\mathcal A$. The defender is therefore the
leader and the attacker is the follower.

For any $z\in\mathcal Z$, define the attacker best response set
\begin{equation}
 \BR_A(z)=\argmax_{a\in\mathcal A}U_A(z,a).
 \label{eq:pure-attacker-br-set}
\end{equation}
Although the defender ultimately selects from $\mathcal X$, we define the response on all
of $\mathcal Z$ because \texttt{RADAR} evaluates infeasible intermediate configurations.
The set $\BR_A(z)$ is nonempty because $\mathcal A$ is finite and nonempty.

When the attacker has multiple best responses, strong tie-breaking selects one that is most
favorable to the defender, while pessimistic tie-breaking selects one that is least
favorable. Define
\begin{equation}
 V_D^+(z)=\max_{a\in\BR_A(z)}U_D(z,a),
 \qquad
 V_D^-(z)=\min_{a\in\BR_A(z)}U_D(z,a).
 \label{eq:pure-stackelberg-values}
\end{equation}
A pure strong Stackelberg equilibrium (SSE) is a pair
$(z^{\mathrm{SE}},a_z^{\mathrm{SE}})$ such that
\begin{equation}
 z^{\mathrm{SE}}\in\argmax_{z\in\mathcal X}V_D^+(z),
 \qquad
 a_z^{\mathrm{SE}}\in
 \argmax_{a\in\BR_A(z^{\mathrm{SE}})}U_D(z^{\mathrm{SE}},a).
 \label{eq:pure-sse-definition}
\end{equation}
The pessimistic pure Stackelberg problem maximizes $V_D^-(z)$ over $z\in\mathcal X$.

\subsection{Attacker Best Response to an Observed Configuration}
\label{subsec:pure-attacker-br}

Fix $z\in\mathcal Z$. Because $\Delta_A(t)>0$, the attacker utility
\eqref{eq:attacker_route_utility} is strictly increasing in $q_r(z)$ among routes to the
same target. The attacker therefore prefers a route with maximum non-detection probability.
Define the nonnegative edge weight
\begin{equation}
 w_e(z)=-\log\bigl(1-p_e(z)\bigr),
 \label{eq:pure-edge-weight}
\end{equation}
with $w_e(z)=+\infty$ when $p_e(z)=1$. For any route with positive non-detection probability,
\begin{equation}
 -\log q_r(z)
 =-\log\prod_{e\in\mathcal E(r)}(1-p_e(z))
 =\sum_{e\in\mathcal E(r)}w_e(z).
 \label{eq:log-route-survival}
\end{equation}
Because $x\mapsto-\log x$ is strictly decreasing on $(0,1]$, maximizing the route
non-detection probability is equivalent to minimizing the sum in
\eqref{eq:log-route-survival}.

To find the best route from any source, add an auxiliary source $s_0$ with a zero-weight
edge to each $s\in\mathcal S$. Assign weight $w_e(z)$ to every original edge and remove
edges with $p_e(z)=1$, which have infinite weight. One Dijkstra run from $s_0$ returns the
minimum distance $d_z(v)$ to every vertex. For target $t$, the maximum route non-detection
probability is
\begin{equation}
 q_t^{\max}(z)=e^{-d_z(t)},
 \label{eq:best-target-survival}
\end{equation}
where $e^{-\infty}=0$. Any predecessor path returned for $t$ maximizes the non-detection
probability when $d_z(t)<\infty$.

The best route to target $t$ gives attacker and defender utility contributions
\begin{align}
 \theta_t(z)
 &=U_A^d(t)+\Delta_A(t)e^{-d_z(t)},
 \label{eq:target-attacker-contribution}\\
 \phi_t(z)
 &=U_D^d(t)-\Delta_D(t)e^{-d_z(t)}.
 \label{eq:target-defender-contribution}
\end{align}
If $\theta_t(z)>0$, every attacker best response includes a route with maximum
non-detection probability to $t$,
because adding it strictly improves attacker utility. If $\theta_t(z)<0$, every best
response omits $t$. When $\theta_t(z)=0$, inclusion and omission give the attacker the same
utility. Strong tie-breaking includes the route exactly when doing so improves the
defender's payoff, while pessimistic tie-breaking includes it exactly when doing so lowers
the defender's payoff. Equivalently,
\begin{align}
 V_D^+(z)
 &=\sum_{t:\,\theta_t(z)>0}\phi_t(z)
   +\sum_{t:\,\theta_t(z)=0}\max\{\phi_t(z),0\},
 \label{eq:pure-strong-value-explicit}\\
 V_D^-(z)
 &=\sum_{t:\,\theta_t(z)>0}\phi_t(z)
   +\sum_{t:\,\theta_t(z)=0}\min\{\phi_t(z),0\}.
 \label{eq:pure-pessimistic-value-explicit}
\end{align}
The particular route chosen among routes of equal transformed weight to the same target does
not affect either player's payoff, because those routes have the same target and the same
non-detection probability.

\begin{theorem}
\label{thm:pure_attacker_br}
For any $z\in\mathcal Z$, the procedure above returns an attacker best response and computes
both $V_D^+(z)$ and $V_D^-(z)$ in
$O((|\mathcal E|+|\mathcal V|)\log|\mathcal V|)$ time.
\end{theorem}
\begin{proof}
Fix a target $t$. By \eqref{eq:attacker_route_utility}, two routes $r,r'\in\mathcal R_t$
satisfy $U_A(z,r)\ge U_A(z,r')$ if and only if $q_r(z)\ge q_{r'}(z)$, because both routes
have the same $U_A^d(t)$ and positive coefficient $\Delta_A(t)$. Equation
\eqref{eq:log-route-survival} shows that a route with maximum non-detection probability is
exactly a minimum weight route under $w(z)$. All finite edge weights are nonnegative, so
Dijkstra's algorithm from
the auxiliary source finds such a route simultaneously for every target. If all routes to a
target contain a perfectly sensed edge, the target distance is infinite and the maximum
non-detection probability is zero, which is represented by $e^{-\infty}=0$.

It remains to form a best action that may attack multiple targets. By
\eqref{eq:action_utilities}, the attacker
utility of an action is the sum of its target contributions, and the constraint in
\eqref{eq:attacker_action_set} couples neither the route choices nor the inclusion decisions
for different targets. For each target, choosing a route with maximum non-detection
probability maximizes its
possible contribution. A positive contribution must be included, a negative contribution
must be omitted, and a zero contribution may be either included or omitted. These
independent choices produce an action in $\BR_A(z)$. At a zero attacker contribution,
choosing between the included route payoff $\phi_t(z)$ and the omitted payoff zero gives the
strong and pessimistic terms in
\eqref{eq:pure-strong-value-explicit}--\eqref{eq:pure-pessimistic-value-explicit}. Therefore
the procedure returns a best response and both defender values.

The auxiliary graph has $|\mathcal V|+1$ vertices and at most
$|\mathcal E|+|\mathcal S|$ edges. With a binary heap, Dijkstra's algorithm takes
$O((|\mathcal E|+|\mathcal V|)\log|\mathcal V|)$ time. Reconstructing at most one route per
target and scanning the targets takes $O(|\mathcal E|+|\mathcal V|)$ total time when
predecessor paths are represented implicitly, so it does not change the stated bound.
\end{proof}

\subsection{Existence and Defender Optimization Hardness}
\label{subsec:pure-se-existence-hardness}

The finite action spaces immediately give existence under either tie-breaking rule.

\begin{theorem}
\label{thm:pure_stackelberg_existence}
A pure SSE exists, and the pessimistic pure Stackelberg problem has an optimal defender
configuration.
\end{theorem}
\begin{proof}
The feasible set $\mathcal X$ is nonempty because it contains the empty configuration, and
it is finite because $\mathcal H$ is finite. For every $z\in\mathcal X$, the finite
best response set $\BR_A(z)$ is nonempty, so the maxima and minima in
\eqref{eq:pure-stackelberg-values} are well defined. A real-valued function on a finite
nonempty set attains both its maximum and minimum. Hence some
$z^+\in\argmax_{z\in\mathcal X}V_D^+(z)$ exists, and selecting a response most favorable to
the defender from $\BR_A(z^+)$ gives a pure SSE\@. Likewise, some
$z^-\in\argmax_{z\in\mathcal X}V_D^-(z)$ exists and solves the pessimistic pure problem.
\end{proof}

The attacker response can be computed in polynomial time, but defender optimization remains
intractable. Consider
the restricted setting with one source, one perfect sensing type of unit cost, no false alarm
or server constraint, and only the cardinality constraint $|z|\le B_D$. Every target has
payoffs
\[
 U_A^u=1,\qquad U_A^d=-1,\qquad
 U_D^u=-1,\qquad U_D^d=1.
\]

\begin{theorem}
\label{thm:defender_pure_se_hard}
Even in the restricted setting above, computing an optimal pure defender configuration is
NP-hard under either strong or pessimistic tie-breaking.
\end{theorem}
\begin{proof}
Appendix~\ref{app:defender-nph-proof} gives a complete reduction from \textsc{Clique} that
is computable in polynomial time. The construction has no attacker ties that affect payoffs,
so the same reduction applies to both tie-breaking rules.
\end{proof}

Theorems~\ref{thm:pure_attacker_br} and~\ref{thm:defender_pure_se_hard} separate evaluation
from optimization: the value of any proposed pure configuration can be computed efficiently,
but finding the best feasible configuration is NP-hard. This motivates a search procedure
that repeatedly evaluates selected configurations without enumerating $\mathcal X$.

\subsection{\texttt{RADAR}: Resource-Aware Downgrading with Adaptive Refinement}
\label{subsec:radar}

\texttt{RADAR} searches $\mathcal Z$ using two stages. Resource-aware downgrading (RAD)
repairs a dense, possibly infeasible configuration by moving sensing actions down prescribed
hierarchies. Adaptive refinement then tests repaired upgrades and exchanges to recover
valuable sensing choices removed during repair.

Let $F:\mathcal Z\to\R$ be the objective used by the search. For pure Stackelberg
optimization, set $F=V_D^+$ under strong tie-breaking and $F=V_D^-$ under pessimistic
tie-breaking. These values are defined on all $\mathcal Z$ because resource feasibility
restricts which configurations may be selected, not how route payoffs are evaluated.
Section~\ref{subsec:rgr} later uses the same search with different objectives.

\subsubsection{Downgrade hierarchies and repair score}

For every edge $e$ that supports sensing, each action $h\in\mathcal H_e$ has a prescribed
immediate downgrade
\[
 d_e(h)\in\mathcal H_e\cup\{\varnothing_e\},
\]
where $\varnothing_e$ denotes no sensing on $e$. A downgrade does not increase detection
probability, deployment cost, server load, or false alarm probability. Repeated downgrades
must reach $\varnothing_e$, so the directed downgrade relation is acyclic. Reversing the
relation defines the immediate upgrades. The initial configuration $z^{\max}$ selects a
designated maximal action on every edge that supports sensing and may violate one or more
resource limits.

If $z$ selects $h$ on edge $e$, let $d_e(z)$ be the configuration obtained by replacing
$h$ with $d_e(h)$, or by removing $h$ when $d_e(h)=\varnothing_e$. Define normalized resource
usage, normalized excess, and the resource reduction from $z$ to $x$ by
\begin{align}
 \mathbf g(z)
 &=\left(\frac{b(z)}{B_D},
          \frac{\ell(z)}{C_{\mathrm{srv}}},
          \frac{P_{\mathrm{fa}}(z)}{\bar P_{\mathrm{fa}}}\right),
 \label{eq:normalized-resource-usage}\\
 \boldsymbol\nu(z)
 &=\pos{\mathbf g(z)-\mathbf 1},
 \label{eq:normalized-resource-excess}\\
 \boldsymbol\xi(z,x)
 &=\mathbf g(z)-\mathbf g(x).
 \label{eq:normalized-resource-reduction}
\end{align}
The positive part operator is applied componentwise. Because a downgrade does not increase
any resource, $\boldsymbol\xi(z,x)\ge\mathbf0$. The amount by which the downgrade addresses
currently violated constraints is
\begin{equation}
 R(z,x)=\boldsymbol\nu(z)^\top\boldsymbol\xi(z,x)\ge0.
 \label{eq:rad-resource-score}
\end{equation}
Let the nonnegative loss in search objective be
\begin{equation}
 \Delta_F(z,x)=\pos{F(z)-F(x)}.
 \label{eq:rad-objective-loss}
\end{equation}
RAD ranks a candidate downgrade by
\begin{equation}
 S_F(z,x)=
 \begin{cases}
 +\infty, & R(z,x)>0\text{ and }\Delta_F(z,x)=0,\\
 R(z,x)/\Delta_F(z,x), & R(z,x)>0\text{ and }\Delta_F(z,x)>0,\\
 0, & R(z,x)=0.
 \end{cases}
 \label{eq:rad_score}
\end{equation}
The score favors a large reduction in the constraints that are currently violated and a
small loss in $F$. Ties are broken first by larger $R(z,x)$ and then by larger $F(x)$.
These secondary rules prefer a downgrade that makes more direct feasibility progress and,
subject to equal progress, retains more objective value.

\subsubsection{Refinement neighborhood}

For a feasible configuration $z$, let $\mathcal U(z)$ be its immediate upgrades. An upgrade
may violate a resource limit, so each upgraded configuration is passed through RAD before it
enters the neighborhood. Let $\mathcal Q(z)\subseteq\mathcal X$ contain feasible exchanges
between two edges. An exchange swaps the sensing type and processing mode choices assigned
to the two edges whenever the resulting actions for those edges belong to $\mathcal H$.
Treating $\varnothing_e$ as an available choice allows an exchange to move a deployment from
a monitored edge to an unmonitored edge. The refinement neighborhood is
\begin{equation}
 \mathcal N_F(z)
 =\{\operatorname{RAD}(y,F):y\in\mathcal U(z)\}\cup\mathcal Q(z).
 \label{eq:radar-neighborhood}
\end{equation}
Duplicate configurations may be removed before evaluation.

\begin{algorithm}[htbp]
\caption{Resource-aware downgrading and adaptive refinement}
\label{alg:radar}
\begin{algorithmic}[1]
\Procedure{RAD}{$z,F$}
\While{$z\notin\mathcal X$}
    \State $\mathcal D\gets\{d_e(z):e\in\mathcal E,\ z_e\ne\emptyset\}$
    \State $z\gets\argmax_{x\in\mathcal D}S_F(z,x)$ using the stated tie rules
\EndWhile
\State \Return $z$
\EndProcedure
\Procedure{\texttt{RADAR}}{$z^{\max},F$}
\State $z\gets\Call{RAD}{z^{\max},F}$
\Loop
    \State $\mathcal N\gets\mathcal N_F(z)$
    \If{$\mathcal N=\emptyset$}
        \State \textbf{break}
    \EndIf
    \State $x\gets\argmax_{y\in\mathcal N}F(y)$
    \If{$F(x)\le F(z)$}
        \State \textbf{break}
    \EndIf
    \State $z\gets x$
\EndLoop
\State \Return $z$
\EndProcedure
\end{algorithmic}
\end{algorithm}

\begin{proposition}
\label{prop:radar_properties}
RAD terminates and returns a feasible configuration. \texttt{RADAR} terminates at a feasible
configuration $z$ for which no candidate in $\mathcal N_F(z)$ has a strictly larger
objective value $F$.
\end{proposition}
\begin{proof}
Assign each action in an edge's downgrade hierarchy a nonnegative rank equal to the number
of remaining downgrade steps before $\varnothing_e$. The sum of the ranks selected by a
configuration is a nonnegative integer. Every RAD iteration replaces one action by its
immediate downgrade and therefore decreases this rank sum by one. RAD cannot cycle and must
terminate after finitely many iterations. If feasibility has not been reached earlier, all
actions are eventually removed. The resulting empty configuration belongs to $\mathcal X$,
so RAD necessarily terminates at a feasible configuration.

The initial RAD call in \texttt{RADAR} therefore returns a feasible configuration. Every
repaired upgrade in
\eqref{eq:radar-neighborhood} is feasible because it is the output of RAD, and every
exchange in $\mathcal Q(z)$ is feasible by definition. Hence every accepted iterate remains
in $\mathcal X$. \texttt{RADAR} accepts a new iterate only when its objective value is strictly
larger. Since $\mathcal X$ is finite, a strictly improving sequence cannot continue
indefinitely and cannot revisit a previous configuration. At termination, either the
neighborhood is empty or its largest value is no larger than $F(z)$, which is exactly the
stated local optimality condition.
\end{proof}

\begin{proposition}
\label{prop:radar-complexity}
Let $M$ be the number of edges that support sensing, let $I$ be the number of refinement
rounds, and let $T_F$ be the time required to evaluate $F$ for one configuration. Assuming
resource usage is maintained incrementally, one RAD call takes
$O(|\mathcal H|M T_F)$ time, and \texttt{RADAR} takes
\begin{equation}
 O\!\left(\bigl(|\mathcal H|M
 +I(|\mathcal H|^2M+M^2)\bigr)T_F\right)
 =O(I|\mathcal H|^3T_F)
 \label{eq:radar-complexity}
\end{equation}
time when $I\ge1$ and $M\le|\mathcal H|$.
\end{proposition}
\begin{proof}
Across all edges, the total number of possible downgrade steps is at most
$|\mathcal H|$: each step leaves one supported action and moves to another supported action
or to no sensing. At any RAD iteration, at most one downgrade is available for each of the
$M$ edges that support sensing, so at most $M$ candidates are scored. Caching $F(z)$ means each
candidate requires at most one new objective evaluation. Thus a repair uses at most
$O(|\mathcal H|M)$ objective evaluations and takes $O(|\mathcal H|M T_F)$ time.

The reverse downgrade relation contains at most $|\mathcal H|$ immediate upgrade arcs in
total, so a refinement round considers at most $|\mathcal H|$ immediate upgrades. Repairing
each one costs $O(|\mathcal H|M T_F)$, for
$O(|\mathcal H|^2M T_F)$ time. There are at most $O(M^2)$ unordered edge pairs to test for
exchanges, each requiring at most a constant number of objective evaluations once support
and feasibility are checked. This contributes $O(M^2T_F)$ per round. Adding the initial
repair and multiplying the refinement cost by $I$ gives the first expression in
\eqref{eq:radar-complexity}. Since $M\le|\mathcal H|$, the expression is
$O(I|\mathcal H|^3T_F)$ for $I\ge1$.
\end{proof}

For pure Stackelberg optimization, Theorem~\ref{thm:pure_attacker_br} gives
\[
 T_F=O((|\mathcal E|+|\mathcal V|)\log|\mathcal V|).
\]
Every local candidate is therefore evaluated after the attacker reroutes optimally. In the
next section, the same repair and refinement operations generate candidate defender columns
for restricted mixed Nash and mixed Stackelberg games; only the objective $F$ changes.

\section{Mixed Stackelberg Game}
\label{sec:mixed_stackelberg}

The preceding Stackelberg setting assumes that the attacker observes the realized pure
defender configuration. The defender may instead use a mixed strategy
$\sigma\in\Delta(\mathcal X)$ and draw one complete feasible configuration before each
interaction. The attacker observes $\sigma$ but not the realized configuration before
selecting its routes. Randomization therefore hides which edges are actively monitored in
the current interaction while preserving the resource feasibility of every realized
configuration.

For route $r$, define its expected non-detection probability under $\sigma$ by
\begin{equation}
 \bar q_r(\sigma)=\sum_{z\in\mathcal X}\sigma(z)q_r(z).
 \label{eq:mixed_route_survival}
\end{equation}
The expected route utilities are
\begin{align}
 \bar U_A(\sigma,r)
 &=U_A^d(t(r))+\Delta_A(t(r))\bar q_r(\sigma),
 \label{eq:mixed-attacker-route-utility}\\
 \bar U_D(\sigma,r)
 &=U_D^d(t(r))-\Delta_D(t(r))\bar q_r(\sigma).
 \label{eq:mixed-defender-route-utility}
\end{align}
They extend additively to attacker actions. The follower best response set is
\begin{equation}
 \BR_A(\sigma)=\argmax_{a\in\mathcal A}\bar U_A(\sigma,a).
 \label{eq:mixed-attacker-br-set}
\end{equation}
Define the mixed defender values under strong and pessimistic tie-breaking by
\begin{equation}
 V_D^+(\sigma)=\max_{a\in\BR_A(\sigma)}\bar U_D(\sigma,a),
 \qquad
 V_D^-(\sigma)=\min_{a\in\BR_A(\sigma)}\bar U_D(\sigma,a).
 \label{eq:mixed-stackelberg-values}
\end{equation}
A mixed SSE is a pair $(\sigma^{\mathrm{SE}},a_\sigma^{\mathrm{SE}})$ satisfying
\begin{equation}
 \sigma^{\mathrm{SE}}\in\argmax_{\sigma\in\Delta(\mathcal X)}V_D^+(\sigma),
 \qquad
 a_\sigma^{\mathrm{SE}}\in
 \argmax_{a\in\BR_A(\sigma^{\mathrm{SE}})}
 \bar U_D(\sigma^{\mathrm{SE}},a).
 \label{eq:mixed-sse-definition}
\end{equation}

The timing is essential. If the realized configuration were revealed before the attacker
selected its action, the defender's expected value from $\sigma$ would be a convex
combination of the Stackelberg values of pure configurations and could not exceed the best
value attained by a pure configuration under the same tie-breaking rule. Randomization can improve the defender's outcome
only because the attacker responds to the distribution rather than to its current draw.

\subsection{Why One Exact Shortest Path Weight No Longer Suffices}
\label{subsec:mixed-no-additive-weight}

For a fixed configuration, the logarithm in \eqref{eq:log-route-survival} converts a path
product into a sum. Under a mixed strategy,
\begin{equation}
 \bar q_r(\sigma)
 =\sum_{z\in\mathcal X}\sigma(z)
   \prod_{e\in\mathcal E(r)}(1-p_e(z)).
 \label{eq:mixed-sum-of-products}
\end{equation}
The logarithm cannot pass through the outer weighted sum. More fundamentally, the common
hidden configuration can correlate the edge states in a way that no fixed set of exact
edge weights represents.

\begin{example}[Correlated hidden configurations]
\label{ex:no-additive-mixed-weight}
Consider two serial binary choice stages, so every source--target route chooses branch
$0$ or $1$ in the first stage and independently chooses branch $0$ or $1$ in the second.
Denote the four routes by $r_{00},r_{01},r_{10},r_{11}$. Let $z^{00}$ leave both branch $0$
choices unsensed and perfectly sense both branch $1$ choices. Let $z^{11}$ do the reverse,
and assign probability $1/2$ to each configuration. Then
\[
 \bar q_{r_{00}}(\sigma)=\bar q_{r_{11}}(\sigma)=\frac12,
 \qquad
 \bar q_{r_{01}}(\sigma)=\bar q_{r_{10}}(\sigma)=0.
\]
Suppose exact nonnegative additive weights existed. Writing the corresponding multiplicative
edge factors as $\alpha_i^b=e^{-w_i^b}\in[0,1]$ for stage $i$ and branch $b$, exactness would
require
\[
 \alpha_1^0\alpha_2^0=\frac12,
 \qquad
 \alpha_1^1\alpha_2^1=\frac12.
\]
Both equalities force all four factors to be positive. The products for the crossed routes
$\alpha_1^0\alpha_2^1$ and $\alpha_1^1\alpha_2^0$ would then also be positive, contradicting
the zero expected non-detection probability of $r_{01}$ and $r_{10}$. Thus no single exact
additive edge weight represents all four expected route non-detection probabilities.
\end{example}

This obstacle is the nonadditive expectation induced by the hidden configuration, not a
failure of Dijkstra's algorithm. It is closely related to reliable path problems with
correlated edge failures \cite{chang2007reachability}.

\subsection{Existence and Pessimistic Nonattainment}
\label{subsec:mixed-existence}

\begin{theorem}
\label{thm:mixed_existence}
A mixed SSE exists. By contrast, the pessimistic mixed Stackelberg problem need
not attain its supremum.
\end{theorem}
\begin{proof}
Appendix~\ref{app:mixed-existence-proof} gives a complete proof. It first partitions the
defender simplex into the closed polytopes on which each attacker action is a best response,
and then maximizes the defender's linear payoff over those polytopes. The appendix also
gives an explicit instance with two configurations in which the pessimistic value approaches a
strictly larger limit from one side and drops at the tie point.
\end{proof}

Every pure configuration is a degenerate mixed strategy. Therefore,
\begin{equation}
\begin{aligned}
 \max_{z\in\mathcal X}V_D^-(z)
 &\le \sup_{\sigma\in\Delta(\mathcal X)}V_D^-(\sigma)\\
 &\le \max_{\sigma\in\Delta(\mathcal X)}V_D^+(\sigma).
\end{aligned}
\label{eq:pure-pess-mixed-strong-order}
\end{equation}
The three quantities are the best pessimistic pure value, the pessimistic mixed supremum,
and the mixed SSE value. Nonattainment is not caused by noncompactness: the simplex
$\Delta(\mathcal X)$ is compact. It is caused by pessimistic tie-breaking, which can make
$V_D^-(\sigma)$ drop discontinuously when a second attacker action becomes tied. Strong
tie-breaking instead selects the tied action most favorable to the defender and attains a
maximum.

\subsection{Explicit Linear Formulation}
\label{subsec:mixed-sse-lp}

Suppose all pure actions are explicitly listed. Fix a candidate attacker response
$a\in\mathcal A$ and solve
\begin{equation}
\begin{aligned}
 \max_{\sigma}\quad
 &\bar U_D(\sigma,a)\\
 \text{s.t.}\quad
 &\bar U_A(\sigma,a)\ge\bar U_A(\sigma,a'),
 &&\forall a'\in\mathcal A,\\
 &\sum_{z\in\mathcal X}\sigma(z)=1,\\
 &\sigma(z)\ge0,
 &&\forall z\in\mathcal X.
\end{aligned}
\label{eq:sse_lp}
\end{equation}
The incentive constraints make $a$ an attacker best response. The objective then evaluates
the defender payoff when strong tie-breaking selects that response.

\begin{proposition}
\label{prop:sse-lp-correctness}
Solving \eqref{eq:sse_lp} for every $a\in\mathcal A$ and retaining a feasible solution with
the largest objective value returns a mixed SSE\@.
\end{proposition}
\begin{proof}
For a fixed $a$, the feasible set of \eqref{eq:sse_lp} consists exactly of the defender
mixtures for which $a\in\BR_A(\sigma)$. Hence every feasible pair $(\sigma,a)$ is a valid
leader strategy together with one follower best response, and its LP objective is the
corresponding defender payoff.

Conversely, for any defender mixture $\sigma$ and any best response
$a\in\BR_A(\sigma)$, the same $\sigma$ is feasible for the LP indexed by $a$. Therefore the
union of the feasible pairs over all indexed LPs is exactly
\[
 \{(\sigma,a):\sigma\in\Delta(\mathcal X),\ a\in\BR_A(\sigma)\}.
\]
Maximizing the defender payoff over this union is precisely the strong Stackelberg
optimization in \eqref{eq:mixed-sse-definition}. Taking the best objective among the
finitely many LPs therefore returns a mixed SSE\@.
\end{proof}

This procedure is polynomial in the size of an explicitly listed game in normal form: it
solves $|\mathcal A|$ linear programs, each with $|\mathcal X|$ variables and
$|\mathcal A|+1$ principal constraints. A compact graph description may nevertheless encode
exponentially many defender configurations and attacker route actions, so the
explicit formulation is not polynomial in the graph input size.

\subsection{Computational Complexity}
\label{subsec:mixed-hardness}

The loss of the shortest path structure is not merely a limitation of the proposed
candidate generator.

\begin{theorem}
\label{thm:mixed_oracle_hard}
In the graph model:
\begin{enumerate}[label=(\roman*),leftmargin=2.2em]
    \item computing an exact attacker best response to a mixed defender strategy with
    explicit finite support is NP-hard, even when the attacker selects only one route;
    \item computing a pure defender best response to a mixed attacker strategy with
    explicit finite support is NP-hard; and
    \item computing an optimal mixed SSE is NP-hard.
\end{enumerate}
\end{theorem}
\begin{proof}
Appendix~\ref{app:mixed-oracle-hard-proof} gives complete reductions from
\textsc{Max-Cut}, \textsc{Maximum Coverage}, and $0$--$1$ \textsc{Knapsack}, respectively.
The constructions use only polynomially many vertices, edges, and explicitly supported
pure strategies, and they use detection probabilities in $\{0,1/2,1\}$. The third
construction has a unique attacker response for every defender mixture, so its hardness is
unrelated to tie-breaking.
\end{proof}

Allowing the defender to randomize therefore makes both exact response generation and
defender optimization combinatorial. Finite formulations remain useful over selected action
sets, but those sets must be expanded without relying on exact oracles that run in
polynomial time.

\subsection{\texttt{RGR}: RADAR-Guided Restricted Game}
\label{subsec:rgr}

\texttt{RGR} maintains nonempty restricted action sets
$\mathcal X'\subseteq\mathcal X$ and $\mathcal A'\subseteq\mathcal A$. Including the empty
configuration and empty attacker action provides a simple feasible initialization, although
the algorithm accepts any nonempty initial sets. For mixed Nash play, it solves the
restricted bimatrix game using the support conditions of
Section~\ref{subsec:ne-support}. For mixed Stackelberg play, it solves the restricted
versions of \eqref{eq:sse_lp}, one for each designated response in $\mathcal A'$, and
retains the best restricted LP value. The resulting solution guides searches for new
defender configurations and attacker actions.

\subsubsection{Defender generation for Nash play}

Let $(\sigma',\mu')$ be a restricted Nash equilibrium. A defender deviation in the full game
would maximize $\bar U_D(z,\mu')$ over $z\in\mathcal X$. \texttt{RGR} therefore calls
\texttt{RADAR} with
\begin{equation}
 F_N(z)=\bar U_D(z,\mu').
 \label{eq:rgr-nash-defender-score}
\end{equation}
The returned configuration is a heuristic pure best response to the restricted attacker
mixture.

\subsubsection{Defender column score for Stackelberg play}

Consider the selected restricted SSE LP with designated attacker action $a^\star$. For
$a'\in\mathcal A'$, define
\[
 d_{a'}(z)=U_A(z,a^\star)-U_A(z,a').
\]
Over $\mathcal X'$, the LP is
\begin{equation}
\begin{aligned}
 \max_{\sigma}\quad
 &\sum_{z\in\mathcal X'}\sigma(z)U_D(z,a^\star)\\
 \text{s.t.}\quad
 &\sum_{z\in\mathcal X'}\sigma(z)d_{a'}(z)\ge0,
 &&\forall a'\in\mathcal A',\\
 &\sum_{z\in\mathcal X'}\sigma(z)=1,
 \qquad \sigma(z)\ge0.
\end{aligned}
\label{eq:restricted-sse-primal}
\end{equation}
Let $\eta_{a'}\ge0$ be the dual multiplier associated with the incentive constraint against
$a'$, and let $\beta$ be the free dual variable associated with the probability equality.
The dual is
\begin{equation}
\begin{aligned}
 \min_{\beta,\eta}\quad &\beta\\
 \text{s.t.}\quad
 &\beta\ge U_D(z,a^\star)
   +\sum_{a'\in\mathcal A'}\eta_{a'}d_{a'}(z),
 &&\forall z\in\mathcal X',\\
 &\eta_{a'}\ge0,
 &&\forall a'\in\mathcal A'.
\end{aligned}
\label{eq:restricted-sse-dual}
\end{equation}
This gives the defender column score
\begin{equation}
 F_S(z)=U_D(z,a^\star)
 +\sum_{a'\in\mathcal A'}\eta_{a'}
 \bigl[U_A(z,a^\star)-U_A(z,a')\bigr].
 \label{eq:rgr-sse-defender-score}
\end{equation}
For every existing column $z\in\mathcal X'$, dual feasibility requires $F_S(z)\le\beta$.
A missing configuration with $F_S(z)>\beta$ violates the current dual constraint and has
positive reduced value for the maximization problem. \texttt{RGR} therefore calls
\texttt{RADAR} with $F_S$ to search for such a column. Pure Stackelberg design, Nash
deviation search, and mixed Stackelberg column generation use the same configuration search;
only $F$ changes.

\subsubsection{Attacker route generation}

The attacker generator combines two candidate sources.

First, the \emph{support generator} applies the exact shortest path response of
Theorem~\ref{thm:pure_attacker_br} to every pure configuration
$z\in\supp(\sigma')$ and collects the routes that appear in those responses. This source is
inexpensive and often useful, but a route that is optimal in expectation need not be optimal
under any individual support configuration.

Second, the \emph{surrogate generator} assigns edge weight
\begin{equation}
 \widetilde w_e(\sigma')
 =\sum_{z\in\supp(\sigma')}\sigma'(z)w_e(z)
 \label{eq:geometric_surrogate}
\end{equation}
and generates the $K_{\mathrm{cand}}$ shortest simple paths to each target. For any route
$r$,
\begin{align*}
 \sum_{e\in\mathcal E(r)}\widetilde w_e(\sigma')
 &=-\sum_{z\in\supp(\sigma')}\sigma'(z)\log q_r(z)\\
 &=-\log\prod_{z\in\supp(\sigma')}q_r(z)^{\sigma'(z)}.
\end{align*}
The surrogate therefore maximizes the weighted geometric mean of the route non-detection
probability, not the expectation in \eqref{eq:mixed_route_survival}. It is used only to
propose paths.
The \emph{combined generator} takes the union of the support and surrogate candidates.

Every proposed route is reevaluated using the exact expectation
\eqref{eq:mixed_route_survival}. For each target, the generator retains a candidate with
maximum exact attacker utility. A target with positive utility is included in the generated
action, and a target with negative utility is omitted. For Nash deviation measurement, a
target with zero utility may be omitted because inclusion does not change attacker utility.
For strong Stackelberg play, zero utility and ties between routes are resolved using the
exact utility most favorable to the defender among the generated candidates.

\begin{algorithm}[htbp]
\caption{\texttt{RGR}: RADAR-guided restricted game}
\label{alg:rgr}
\begin{algorithmic}[1]
\Require Game $g\in\{\mathrm{NE},\mathrm{SSE}\}$; nonempty
$\mathcal X',\mathcal A'$; candidate count $K_{\mathrm{cand}}$; iteration cap
\Loop
    \State Solve the restricted game
    \If{$g=\mathrm{NE}$}
        \State obtain $(\sigma',\mu')$ and set $F\gets F_N$
    \Else
        \State obtain $(\sigma',a^\star,\eta,\beta)$ and set $F\gets F_S$
    \EndIf
    \State $z^+\gets\Call{\texttt{RADAR}}{z^{\max},F}$
    \State generate support and surrogate routes and form action $a^+$ using exact utilities
    \State add any new generated defender or attacker action that improves or invalidates
    the restricted solution
    \If{no action is added or the iteration cap is reached}
        \State \textbf{break}
    \EndIf
\EndLoop
\State \Return the restricted solution and generated action diagnostics
\end{algorithmic}
\end{algorithm}

For a restricted Nash solution, the generated actions give the measured unilateral
deviation gains
\begin{align}
 \epsilon_D
 &=\pos{\bar U_D(z^+,\mu')-
          \bar U_D(\sigma',\mu')},
 \label{eq:rgr-nash-defender-gap}\\
 \epsilon_A
 &=\pos{\bar U_A(\sigma',a^+)-
          \bar U_A(\sigma',\mu')}.
 \label{eq:rgr-nash-attacker-gap}
\end{align}
For the selected restricted SSE LP, the corresponding generated quantities are
\begin{align}
 \epsilon_{\mathrm{col}}
 &=\pos{F_S(z^+)-\beta},
 \label{eq:rgr-sse-column-gap}\\
 \epsilon_{\mathrm{route}}
 &=\pos{\bar U_A(\sigma',a^+)-
          \bar U_A(\sigma',a^\star)}.
 \label{eq:rgr-sse-route-gap}
\end{align}
The evaluator's generated action gaps summarize improvements found among generated actions;
they are not global equilibrium gaps.

\begin{proposition}
\label{prop:rgr-certificates}
If exact best response oracles for the full game are used for both players and a restricted
Nash solution has $\epsilon_D=\epsilon_A=0$, then it is a Nash equilibrium of the full game.
A full mixed SSE certificate additionally requires considering every possible designated attacker
response and pricing every corresponding LP over all defender configurations.
The heuristic \texttt{RGR} procedure does not provide these global certificates in general.
\end{proposition}
\begin{proof}
For Nash play, zero exact defender and attacker gains mean that neither player has a
profitable pure deviation in its full action set. Linearity then rules out profitable mixed
deviations, so the restricted solution satisfies the full Nash conditions
\eqref{eq:mixed-ne-defender}--\eqref{eq:mixed-ne-attacker}.

For mixed SSE, a zero exact pricing violation certifies only the currently selected LP and
its designated response. The global strong Stackelberg problem is the maximum over one LP
indexed by each $a\in\mathcal A$, as shown in Proposition~\ref{prop:sse-lp-correctness}. A missing
designated response may define a different LP with a larger defender value, and a missing
attacker action may invalidate the incentive constraints of an existing LP\@. Therefore a
global certificate requires complete response enumeration and complete pricing for every
indexed LP\@. Theorem~\ref{thm:mixed_oracle_hard} shows why \texttt{RGR} uses heuristic
generation instead.
\end{proof}

Because the action sets are finite and only new actions are added, \texttt{RGR} terminates
at the iteration cap or when its generators add no action. Under heuristic generation, this
stopping condition does not certify global optimality.

\section{Relations Between Nash and Stackelberg Equilibria}
\label{sec:relations}

The preceding games have the same action spaces and payoffs but differ in move order and in
what the attacker observes. Existing equivalence results for target security games
\cite{korzhyk2011nash} do not directly cover multiplicative route detection and
coordinated attacks against multiple targets. We first characterize when a pure
Stackelberg outcome is also a pure Nash equilibrium. We then show that a mixed SSE weakly improves the defender's mixed Nash payoff, characterize
the attacker's minimax value when each action contains at most one route, extend that
characterization under a common payoff ratio, and identify the weighted zero-sum case in which
mixed Nash and mixed Stackelberg payoffs coincide.

\subsection{When a Pure Stackelberg Outcome Is a Pure Nash Equilibrium}
\label{subsec:pure-se-ne-relation}

A pure Stackelberg outcome $(z^{\mathrm{SE}},a_z^{\mathrm{SE}})$ already satisfies the
attacker's Nash condition because $a_z^{\mathrm{SE}}\in\BR_A(z^{\mathrm{SE}})$. Only the
defender condition remains.

\begin{proposition}
\label{prop:pure-se-is-ne}
A pure Stackelberg outcome $(z^{\mathrm{SE}},a_z^{\mathrm{SE}})$ is a pure Nash equilibrium
if and only if
\begin{equation}
 z^{\mathrm{SE}}
 \in\argmax_{z\in\mathcal X}U_D(z,a_z^{\mathrm{SE}}).
 \label{eq:pure-se-ne-condition}
\end{equation}
\end{proposition}
\begin{proof}
Because $a_z^{\mathrm{SE}}$ is an attacker best response to $z^{\mathrm{SE}}$, the attacker
inequality \eqref{eq:pure-ne-attacker} is satisfied. The pair is therefore a pure NE exactly
when it also satisfies the defender inequality \eqref{eq:pure-ne-defender}. That inequality
is equivalent to \eqref{eq:pure-se-ne-condition}.
\end{proof}

Condition \eqref{eq:pure-se-ne-condition} need not hold. In the instance used to prove
Theorem~\ref{thm:pure-ne-existence}, a pure Stackelberg optimum is $z_1$ followed by
$\{r_2\}$, or symmetrically $z_2$ followed by $\{r_1\}$. The defender obtains $-1$ at either
outcome. Holding the attacker's action fixed, however, the defender can switch to the other
sensor placement and obtain $1$. Thus a pure Stackelberg outcome can fail to be a pure NE
even though the attacker is best responding.

\subsection{General Defender Advantage under Strong Tie-Breaking}
\label{subsec:leader-advantage}

Under strong tie-breaking, a mixed SSE gives the defender at least the payoff of any mixed
Nash equilibrium, without requiring a zero-sum or common ratio assumption.

\begin{theorem}
\label{thm:mixed-sse-defender-dominates-ne}
Let $(\sigma^{\mathrm{SE}},a_\sigma^{\mathrm{SE}})$ be a mixed SSE and let
$(\sigma^{\mathrm{NE}},\mu^{\mathrm{NE}})$ be any mixed NE\@. Then
\begin{equation}
 \bar U_D(\sigma^{\mathrm{SE}},a_\sigma^{\mathrm{SE}})
 \ge
 \bar U_D(\sigma^{\mathrm{NE}},\mu^{\mathrm{NE}}).
 \label{eq:mixed-sse-defender-dominates-ne}
\end{equation}
\end{theorem}
\begin{proof}
Every action in $\supp(\mu^{\mathrm{NE}})$ is an attacker best response to
$\sigma^{\mathrm{NE}}$. Strong tie-breaking at that same defender mixture gives
\begin{align*}
 V_D^+(\sigma^{\mathrm{NE}})
 &=\max_{a\in\BR_A(\sigma^{\mathrm{NE}})}
   \bar U_D(\sigma^{\mathrm{NE}},a)\\
 &\ge \sum_{a\in\supp(\mu^{\mathrm{NE}})}
   \mu^{\mathrm{NE}}(a)\bar U_D(\sigma^{\mathrm{NE}},a)\\
 &=\bar U_D(\sigma^{\mathrm{NE}},\mu^{\mathrm{NE}}).
\end{align*}
The mixed SSE maximizes $V_D^+(\sigma)$ over all defender mixtures, so its payoff is at
least $V_D^+(\sigma^{\mathrm{NE}})$ and hence at least the mixed Nash payoff.
\end{proof}

The theorem orders defender payoffs but does not order attacker payoffs. Additional payoff
structure is needed for that comparison.

\subsection{Minimax Value under Single Route Attacks}
\label{subsec:single-route-minimax}

Suppose every nonempty attacker action contains exactly one route, so the attacker chooses
at most one target. Target variation in defender and attacker payoffs can then be absorbed by
reweighting the attacker's route mixture.

\begin{lemma}
\label{lem:single_route-bound}
Assume every attacker action contains at most one route. For any mixed NE
$(\sigma^{\mathrm{NE}},\mu^{\mathrm{NE}})$, any
$\sigma\in\Delta(\mathcal X)$, and any $a\in\BR_A(\sigma)$,
\begin{equation}
\begin{aligned}
 \bar U_A(\sigma^{\mathrm{NE}},\mu^{\mathrm{NE}})
 &=\min_{\sigma'\in\Delta(\mathcal X)}
   \max_{a'\in\mathcal A}\bar U_A(\sigma',a')\\
 &\le \bar U_A(\sigma,a).
\end{aligned}
\label{eq:single_route_bound}
\end{equation}
Thus the attacker's mixed Nash payoff is its minimax value and does not exceed its payoff
from a best response to any pure or mixed defender strategy.
\end{lemma}
\begin{proof}
Appendix~\ref{app:single-route-proof} gives the complete reweighting and saddle point
argument, including the separate case in which the empty action belongs to the Nash support.
\end{proof}

The lemma compares values rather than arbitrary defender strategies: an optimal
Stackelberg defense need not itself minimize the attacker's payoff. With explicitly listed
actions, a minimax defender mixture is obtained from
\begin{equation}
\begin{aligned}
 \min_{\sigma,v}\quad &v\\
 \text{s.t.}\quad
 &\sum_{z\in\mathcal X}\sigma(z)U_A(z,a)\le v,
 &&\forall a\in\mathcal A,\\
 &\sum_{z\in\mathcal X}\sigma(z)=1,\\
 &\sigma(z)\ge0,
 &&\forall z\in\mathcal X.
\end{aligned}
\label{eq:attacker-minimax-lp}
\end{equation}
Recovering a Nash attacker mixture for the original general-sum game additionally requires
a saddle point mixture for the zero-sum game defined by the attacker payoff and the inverse
reweighting in Appendix~\ref{app:single-route-proof}. \texttt{RGR} instead solves the restricted original
game and remains applicable when one attacker action contains several routes.

\subsection{Multiple Routes and a Common Payoff Ratio}
\label{subsec:common-ratio}

When actions contain several routes, reweighting probabilities separately by target need not
produce one valid attacker mixture. A payoff ratio that is constant across targets restores a
single identity that makes the game strategically zero sum.

\begin{theorem}
\label{thm:common_ratio}
Suppose a constant $\lambda>0$ satisfies
\begin{equation}
 \frac{\Delta_D(t)}{\Delta_A(t)}=\lambda,
 \qquad \forall t\in\mathcal T.
 \label{eq:common_ratio}
\end{equation}
When an attacker action may contain any number of routes, every mixed NE gives the attacker
the minimax value in
\eqref{eq:single_route_bound}. Consequently, the attacker's payoff from a best response to
any pure or mixed defender strategy is at least its mixed Nash payoff.
\end{theorem}
\begin{proof}
For each target, define
\begin{equation}
 c_t=U_D^u(t)+\lambda U_A^u(t).
 \label{eq:target-action-constant}
\end{equation}
Condition \eqref{eq:common_ratio} implies
\begin{align*}
 U_D^d(t)+\lambda U_A^d(t)
 &=U_D^u(t)+\Delta_D(t)
   +\lambda\bigl(U_A^u(t)-\Delta_A(t)\bigr)\\
 &=U_D^u(t)+\lambda U_A^u(t)
   +\Delta_D(t)-\lambda\Delta_A(t)\\
 &=c_t.
\end{align*}
Using either the detected/undetected expectation or the affine forms
\eqref{eq:defender_route_utility}--\eqref{eq:attacker_route_utility}, every route to target
$t$ therefore satisfies
\[
 U_D(z,r)+\lambda U_A(z,r)=c_t.
\]
For an attacker action $a$, let
\begin{equation}
 C(a)=\sum_{r\in a}c_{t(r)}.
 \label{eq:action-constant}
\end{equation}
Additivity gives the outcome identity
\begin{equation}
 U_D(z,a)+\lambda U_A(z,a)=C(a),
 \qquad \forall z\in\mathcal X,\ a\in\mathcal A.
 \label{eq:strategic_zero_sum}
\end{equation}
The right-hand side depends on the attacker action but not on the defender action.

Let $(\sigma^{\mathrm{NE}},\mu^{\mathrm{NE}})$ be a mixed NE\@. For a fixed attacker mixture
$\mu^{\mathrm{NE}}$, taking expectations in \eqref{eq:strategic_zero_sum} yields
\[
 \bar U_D(\sigma,\mu^{\mathrm{NE}})
 =\mathbb E_{a\sim\mu^{\mathrm{NE}}}[C(a)]
  -\lambda\bar U_A(\sigma,\mu^{\mathrm{NE}}).
\]
The first term is independent of $\sigma$. Because $\sigma^{\mathrm{NE}}$ maximizes defender
utility against $\mu^{\mathrm{NE}}$, it minimizes attacker utility against that same
mixture. The attacker Nash condition says that $\mu^{\mathrm{NE}}$ maximizes attacker
utility against $\sigma^{\mathrm{NE}}$. Hence
$(\sigma^{\mathrm{NE}},\mu^{\mathrm{NE}})$ is a saddle point of the zero-sum game in which
the defender minimizes and the attacker maximizes $\bar U_A$. Its attacker payoff is
therefore
\[
 \min_{\sigma'\in\Delta(\mathcal X)}
 \max_{a'\in\mathcal A}\bar U_A(\sigma',a').
\]
For any $\sigma$ and any $a\in\BR_A(\sigma)$,
$\bar U_A(\sigma,a)=\max_{a'}\bar U_A(\sigma,a')$, which is at least the minimax value.
This proves both statements.
\end{proof}

Theorems~\ref{thm:mixed-sse-defender-dominates-ne} and~\ref{thm:common_ratio} give
nonnegative defender and attacker payoff differences between mixed SSE and mixed NE\@. The
range of $C(a)$ bounds those differences.

\begin{theorem}
\label{thm:common_ratio_gap}
Let
\[
 C_{\min}=\min_{a\in\mathcal A}C(a),
 \qquad
 C_{\max}=\max_{a\in\mathcal A}C(a).
\]
Under \eqref{eq:common_ratio}, a mixed SSE
$(\sigma^{\mathrm{SE}},a_\sigma^{\mathrm{SE}})$ and any mixed NE
$(\sigma^{\mathrm{NE}},\mu^{\mathrm{NE}})$ satisfy
\begin{align}
 0&\le
 \bar U_D(\sigma^{\mathrm{SE}},a_\sigma^{\mathrm{SE}})
 -\bar U_D(\sigma^{\mathrm{NE}},\mu^{\mathrm{NE}})
 \le C_{\max}-C_{\min},
 \label{eq:defender_gap}\\
 0&\le
 \bar U_A(\sigma^{\mathrm{SE}},a_\sigma^{\mathrm{SE}})
 -\bar U_A(\sigma^{\mathrm{NE}},\mu^{\mathrm{NE}})
 \le \frac{C_{\max}-C_{\min}}{\lambda}.
 \label{eq:attacker_gap}
\end{align}
\end{theorem}
\begin{proof}
Let $\Delta_D^{\mathrm{pay}}$ and $\Delta_A^{\mathrm{pay}}$ denote the defender and attacker
payoff differences on the left sides of \eqref{eq:defender_gap} and
\eqref{eq:attacker_gap}. Theorem~\ref{thm:mixed-sse-defender-dominates-ne} gives
$\Delta_D^{\mathrm{pay}}\ge0$, and Theorem~\ref{thm:common_ratio} gives
$\Delta_A^{\mathrm{pay}}\ge0$.

Apply \eqref{eq:strategic_zero_sum} to the SSE outcome and take its expectation at the NE:
\begin{align*}
 &\bar U_D(\sigma^{\mathrm{SE}},a_\sigma^{\mathrm{SE}})
 +\lambda\bar U_A(\sigma^{\mathrm{SE}},a_\sigma^{\mathrm{SE}})
 =C(a_\sigma^{\mathrm{SE}}),\\
 &\bar U_D(\sigma^{\mathrm{NE}},\mu^{\mathrm{NE}})
 +\lambda\bar U_A(\sigma^{\mathrm{NE}},\mu^{\mathrm{NE}})
 =\mathbb E_{a\sim\mu^{\mathrm{NE}}}[C(a)].
\end{align*}
Subtracting gives
\begin{equation}
 \Delta_D^{\mathrm{pay}}+\lambda\Delta_A^{\mathrm{pay}}
 =C(a_\sigma^{\mathrm{SE}})
  -\mathbb E_{a\sim\mu^{\mathrm{NE}}}[C(a)].
 \label{eq:sum-of-payoff-gaps}
\end{equation}
Both quantities on the right lie in $[C_{\min},C_{\max}]$, so their difference is at most
$C_{\max}-C_{\min}$. Since the two terms on the left are nonnegative, each is individually
bounded by that range, with the attacker term divided by $\lambda$. This yields
\eqref{eq:defender_gap}--\eqref{eq:attacker_gap}.
\end{proof}

\subsection{Weighted Zero-Sum Case}
\label{subsec:weighted-zero-sum}

The gap bounds become zero when $C(a)$ is the same for every attacker action. Because the
empty action belongs to $\mathcal A$ and $C(\emptyset)=0$, the common value must be zero.
Since singleton route actions are available, this requires $c_t=0$ for every target. Together
with \eqref{eq:common_ratio}, this is equivalent to
\begin{equation}
 U_D^u(t)=-\lambda U_A^u(t),
 \qquad
 U_D^d(t)=-\lambda U_A^d(t),
 \qquad \forall t\in\mathcal T.
 \label{eq:weighted_zero_sum}
\end{equation}

\begin{theorem}
\label{thm:weighted-zero-sum-relations}
Assume \eqref{eq:weighted_zero_sum}. Let
$(z^{\mathrm{SE}},a_z^{\mathrm{SE}})$ be an optimal pure Stackelberg outcome,
$(\sigma^{\mathrm{SE}},a_\sigma^{\mathrm{SE}})$ a mixed SSE, and
$(\sigma^{\mathrm{NE}},\mu^{\mathrm{NE}})$ any mixed NE\@. Then strong and pessimistic
tie-breaking give the same defender value, and
\begin{align}
 U_D(z^{\mathrm{SE}},a_z^{\mathrm{SE}})
 &\le
 \bar U_D(\sigma^{\mathrm{SE}},a_\sigma^{\mathrm{SE}})
 =\bar U_D(\sigma^{\mathrm{NE}},\mu^{\mathrm{NE}}),
 \label{eq:weighted-zero-sum-defender-order}\\
 U_A(z^{\mathrm{SE}},a_z^{\mathrm{SE}})
 &\ge
 \bar U_A(\sigma^{\mathrm{SE}},a_\sigma^{\mathrm{SE}})
 =\bar U_A(\sigma^{\mathrm{NE}},\mu^{\mathrm{NE}}).
 \label{eq:weighted-zero-sum-attacker-order}
\end{align}
Moreover, the defender payoff gain from randomization satisfies
\begin{equation}
\begin{aligned}
 &\bar U_D(\sigma^{\mathrm{SE}},a_\sigma^{\mathrm{SE}})
 -U_D(z^{\mathrm{SE}},a_z^{\mathrm{SE}})\\
 &\quad=\lambda\left[
 U_A(z^{\mathrm{SE}},a_z^{\mathrm{SE}})
 -\bar U_A(\sigma^{\mathrm{SE}},a_\sigma^{\mathrm{SE}})
 \right].
\end{aligned}
\label{eq:randomization_gain}
\end{equation}
\end{theorem}
\begin{proof}
Equation \eqref{eq:weighted_zero_sum} gives
\begin{equation}
 U_D(z,a)=-\lambda U_A(z,a)
 \label{eq:outcome-weighted-zero-sum}
\end{equation}
for every pure outcome, and hence the same identity for all mixed outcomes. If two attacker
actions are tied in attacker utility against a defender strategy, they are also tied in
defender utility by \eqref{eq:outcome-weighted-zero-sum}. Strong and pessimistic
tie-breaking therefore coincide.

Let
\[
 v_A=\min_{\sigma\in\Delta(\mathcal X)}
      \max_{a\in\mathcal A}\bar U_A(\sigma,a)
\]
be the attacker value of the zero-sum game. Every mixed NE has attacker payoff $v_A$ and
defender payoff $-\lambda v_A$. For a defender mixture $\sigma$, every attacker best
response has utility $\max_a\bar U_A(\sigma,a)$, so the defender's mixed Stackelberg value
at $\sigma$ is
\[
 -\lambda\max_{a\in\mathcal A}\bar U_A(\sigma,a).
\]
Maximizing this expression over $\sigma$ is equivalent to minimizing the inner maximum.
Thus the mixed SSE has attacker payoff $v_A$ and defender payoff $-\lambda v_A$, proving
the equalities in \eqref{eq:weighted-zero-sum-defender-order}--
\eqref{eq:weighted-zero-sum-attacker-order}.

A pure defender configuration is a degenerate mixed strategy, so restricting the defender
to pure strategies cannot improve its maximum payoff. This proves the defender inequality
in \eqref{eq:weighted-zero-sum-defender-order}. Multiplying that inequality by
$-1/\lambda$ and using \eqref{eq:outcome-weighted-zero-sum} gives the attacker inequality
in \eqref{eq:weighted-zero-sum-attacker-order}. Finally, subtracting
$U_D=-\lambda U_A$ at the pure and mixed Stackelberg outcomes yields
\eqref{eq:randomization_gain} exactly.
\end{proof}

Without the common ratio condition, a mixed SSE still weakly improves the defender's mixed
Nash payoff by Theorem~\ref{thm:mixed-sse-defender-dominates-ne}, but the attacker payoff has
no universal ordering. The evaluation uses the general model in which payoffs vary by target and then checks the
specialized identities separately on small constructed profiles.

\section{Evaluation}
\label{sec:evaluation}

We evaluate four questions. First, how does the exact shortest path response to an observed
pure configuration compare with route generation against a mixed defender strategy? Second,
how closely does \texttt{RADAR} approximate exact optimization over pure configurations, and
how much does adaptive refinement improve over RAD repair alone? Third, how closely do
\texttt{RGR}'s restricted Nash and mixed Stackelberg solutions match fully enumerated
references, and what diagnostics remain when exhaustive enumeration is unavailable? Fourth,
do small constructed instances satisfy the payoff relationships in
Section~\ref{sec:relations}?

The evaluation contains 51,929 recorded rows. We distinguish exact comparisons from
diagnostics over generated actions: a zero generated action gap is not a global certificate
unless the corresponding exact oracle was used.

\subsection{Experimental Setup}
\label{subsec:evaluation-setup}

\paragraph{Graph families and instance counts.}
The enumerable experiments use diamond graphs (DGs) with
$|\mathcal V|\in\{4,7\}$ and random geometric graphs (RGGs) with
$|\mathcal V|\in\{6,7,8\}$. Diamond graphs provide controlled alternative routes, while
RGGs provide irregular topology. Each topology and size contains 100 independently
parameterized instances. The restricted game records also include Diamond-10, RGG-9, and
RGG-10. Full enumeration is unavailable for these larger settings, so we report completion
rates and generated action diagnostics but not exact route or defender payoff comparisons.

For optimization over pure configurations, the baseline resource vector
$(B_D,C_{\mathrm{srv}},\bar P_{\mathrm{fa}})$ and separate $0.75\times$ and $1.25\times$
changes to each component produce seven resource cases per instance. This gives 700 cases
per enumerable topology and 3,500 cases overall. We pool the seven cases for each topology.

\paragraph{Sensing parameters.}
Each edge supports three sensing actions. Their base detection probabilities, false alarm
probabilities, deployment costs, and server loads are, respectively,
\begin{align*}
 &\{0.30,0.58,0.76\},
 &&\{0.0015,0.0035,0.0055\},\\
 &\{0.75,1.20,1.35\},
 &&\{0,0,0.85\}.
\end{align*}
An edge perturbation $\delta_e$ is added to each base detection probability, with
$\delta_e\sim\operatorname{Unif}[-0.035,0.035]$ for DGs and
$\delta_e\sim\operatorname{Unif}[-0.045,0.045]$ for RGGs. Across both graph families, the
three resulting detection ranges are $[0.255,0.345]$, $[0.535,0.625]$, and
$[0.715,0.805]$.

\paragraph{Implementation.}
The implementation is written in C++ and uses Gurobi for the finite optimization problems.
Experiments run on an AMD Ryzen 9 9950X3D workstation with 96~GB of memory. The greedy
baseline starts with no sensing and repeatedly selects the feasible single upgrade with the
largest increase in defender utility, stopping when none improves the objective. Reported
runtimes are arithmetic means per case. The reporter treats differences within $10^{-8}$ as
numerically exact and displays every value with magnitude below $10^{-12}$ as $0$ throughout
this section and Appendix~\ref{app:detailed-evaluation}. The path enumerator also records
whether it reaches its candidate cap; no reported row does.

\paragraph{Reproducibility boundary of the supplied records.}
The supplied records do not specify the RGG connection radius, source and target selection
rule, payoff generation rule, baseline resource vector, initial restricted action sets, or
\texttt{RGR} iteration cap. We do not infer these settings; they should be added from the
experiment configuration before public release. All numerical statements below use only the
available records.

\subsection{Metrics}
\label{subsec:evaluation-metrics}

For a mixed defender strategy with an exhaustively enumerable attacker action set, let
$a^+$ be the generated action and define the route response gap
\begin{equation}
 g_A^{\mathrm{route}}
 =\max_{a\in\mathcal A}\bar U_A(\sigma,a)
  -\bar U_A(\sigma,a^+)\ge0.
 \label{eq:evaluation-route-gap}
\end{equation}
The evaluator stores the absolute numerical value of this gap. Checks for an observed pure
configuration compare both non-detection probability and utility with exhaustive enumeration,
but do not store a numerical gap once the check passes.

For optimization over pure configurations, let $V_D^\star$ be the exact optimum under the
selected tie-breaking rule and let $V_D$ be the method's value. Let $R_D$ be the defender
payoff range over the enumerated reference actions. The normalized defender gap is
\begin{equation}
 g_D=\frac{V_D^\star-V_D}{R_D}.
 \label{eq:evaluation-pure-defender-gap}
\end{equation}
An exact match means that this difference is within the reporter tolerance.

For restricted mixed games, define the signed normalized payoff difference
\begin{equation}
 \delta_D=\frac{V_D^{\mathrm{full}}-V_D^{\texttt{RGR}}}{R_D}.
 \label{eq:evaluation-rgr-defender-gap}
\end{equation}
For Nash play, $V_D^{\mathrm{full}}$ and $V_D^{\texttt{RGR}}$ are payoffs from equilibria
selected by the solver. Their difference is not an optimality gap when either game has
multiple equilibria. For mixed Stackelberg play, the values compare selected SSEs of the
fully enumerated and restricted games, but omitted attacker actions can make the restricted
value optimistic or pessimistic. We therefore report both signed and absolute differences
where available.

\subsection{Attacker Response: Pure Configuration versus Mixed Strategy}
\label{subsec:evaluation-attacker-response}

The test for an observed pure configuration covers 500 baseline instances: 100 each for Diamond-4,
Diamond-7, RGG-6, RGG-7, and RGG-8. The Dijkstra response passes exhaustive checks of
non-detection probability and utility in all 500 cases, confirming the construction in
Theorem~\ref{thm:pure_attacker_br} on the tested instances.

For mixed defender strategies, the \emph{support} generator applies Dijkstra to each
configuration in $\supp(\sigma)$, the \emph{surrogate} generator produces the
$K_{\mathrm{cand}}$ shortest paths under \eqref{eq:geometric_surrogate}, and the
\emph{combined} generator uses their union. Every candidate is reevaluated under the exact
expectation. Table~\ref{tab:route-generator-aggregate} aggregates the enumerable Nash and
mixed SSE records.

\begin{table}[htbp]
\centering
\caption{Route generation against mixed defender strategies over enumerable cases.
``Exact'' is the fraction within $10^{-8}$ of exhaustive attacker utility. Support records
are pooled across the three candidate counts because that generator ignores
$K_{\mathrm{cand}}$.}
\label{tab:route-generator-aggregate}
\begin{adjustbox}{max width=\textwidth}
\begin{tabular}{lrrrr}
\toprule
Candidate source & Exact comparisons & Exact (\%) & Mean route gap & Max route gap \\
\midrule
Support & 3000 & 53.8 & $1.5131\times10^{-2}$ & $2.543\times10^{-1}$ \\
Surrogate, $K_{\mathrm{cand}}=1$ & 1000 & 92.9 & $1.3840\times10^{-3}$ & $2.162\times10^{-1}$ \\
Surrogate, $K_{\mathrm{cand}}=4$ & 1000 & 96.4 & $5.6274\times10^{-4}$ & $7.554\times10^{-2}$ \\
Surrogate, $K_{\mathrm{cand}}=8$ & 1000 & 98.6 & $2.6848\times10^{-4}$ & $7.554\times10^{-2}$ \\
Combined, $K_{\mathrm{cand}}=1$ & 1000 & 95.0 & $6.7948\times10^{-4}$ & $7.968\times10^{-2}$ \\
Combined, $K_{\mathrm{cand}}=4$ & 1000 & 97.5 & $2.4170\times10^{-4}$ & $4.319\times10^{-2}$ \\
Combined, $K_{\mathrm{cand}}=8$ & 1000 & 99.2 & $6.0360\times10^{-5}$ & $1.761\times10^{-2}$ \\
\bottomrule
\end{tabular}
\end{adjustbox}
\end{table}

The support generator shows why best responses to individual pure configurations are
insufficient for a hidden mixed defense: only 53.8\% of the pooled records contain an exact
expected utility best response among those routes. The surrogate recovers routes that are not
optimal for any support configuration. Increasing $K_{\mathrm{cand}}$ from 1 to 8 raises its
exact rate from 92.9\% to 98.6\% and reduces its mean route gap by about a factor of five. The
combined generator performs best at every tested candidate count; at
$K_{\mathrm{cand}}=8$, it matches 992 of 1000 exhaustive responses and has mean gap
$6.04\times10^{-5}$.

The remaining misses are concentrated in RGG-8. With combined
$K_{\mathrm{cand}}=8$, all Diamond-4, Diamond-7, RGG-6, and RGG-7 attacker responses are
exact for both Nash and mixed SSE strategies. RGG-8 is exact in 96 of 100 Nash cases and 96
of 100 mixed SSE cases, with mean gaps $2.627\times10^{-4}$ and
$3.409\times10^{-4}$, respectively. These results support the use of the combined generator
but do not constitute an approximation bound; Theorem~\ref{thm:mixed_oracle_hard} rules out
a general exact oracle that runs in polynomial time unless $\mathrm{P}=\mathrm{NP}$.

\subsection{Optimization over Pure Defender Configurations}
\label{subsec:evaluation-radar}

All four methods return feasible configurations in every one of the 3,500 cases under both
strong and pessimistic tie-breaking. The exact match rates and normalized gaps are identical
under the two tie-breaking rules in the supplied records; mean runtimes differ only slightly.
Table~\ref{tab:defender-eval-overall} uses the pessimistic runtimes and pools the five
topology settings, each of which contributes 700 cases.

\begin{table}[htbp]
\centering
\caption{Optimization over pure configurations in 3,500 cases. Gaps are percentages of the
enumerated defender payoff range $R_D$.}
\label{tab:defender-eval-overall}
\begin{adjustbox}{max width=\textwidth}
\begin{tabular}{lrrrrr}
\toprule
Method & Exact match (\%) & Mean gap (\%) & Max gap (\%) & Time (ms) & Objective evals \\
\midrule
Exact enumeration & 100.0 & 0 & 0 & 19.020 & 15912.5 \\
\texttt{RADAR} & 47.3 & 3.640 & 48.36 & 0.329 & 242.2 \\
RAD & 23.1 & 8.714 & 75.47 & 0.189 & 164.8 \\
Greedy & 14.4 & 42.852 & 77.99 & 0.0265 & 19.1 \\
\bottomrule
\end{tabular}
\end{adjustbox}
\end{table}

Adaptive refinement reduces RAD's mean gap from 8.71\% to 3.64\% and raises its exact match
rate from 23.1\% to 47.3\%. The additional objective evaluations therefore recover
substantial payoff that feasibility repair alone discards. Greedy is faster but has a
42.85\% mean gap, showing that locally adding sensing from the empty configuration is much
less effective than repairing a dense configuration and then reconsidering removed or
misplaced sensing actions.

Table~\ref{tab:defender-eval-topology} separates the topology settings. The exact enumerator
is slightly faster than \texttt{RADAR} on Diamond-4 because the complete action space is
tiny: it evaluates about 34 configurations, compared with about 42 \texttt{RADAR}
evaluations. From Diamond-7 onward, exact enumeration is 8.7--116.5 times slower and uses
9.8--137.9 times as many objective evaluations. The largest contrast occurs on RGG-8,
where exact enumeration averages 56,810 evaluations and 68.6~ms, while \texttt{RADAR}
averages 412 evaluations and 0.589~ms.

\begin{table}[htbp]
\centering
\caption{Pure configuration results by topology under pessimistic tie-breaking. Strong
tie-breaking gives the same quality metrics.}
\label{tab:defender-eval-topology}
\begin{adjustbox}{max width=\textwidth}
\begin{tabular}{lrrrrrrrrr}
\toprule
& \multicolumn{5}{c}{\texttt{RADAR}} & \multicolumn{2}{c}{Exact} & \multicolumn{2}{c}{Mean gap (\%)}\\
Topology & Exact (\%) & Mean gap (\%) & Max gap (\%) & Evals & Time (ms) & Evals & Time (ms) & RAD & Greedy \\
\midrule
Diamond-4 & 82.1 & 0.221 & 3.338 & 41.8 & 0.0510 & 33.6 & 0.0409 & 0.498 & 35.81 \\
Diamond-7 & 22.7 & 2.023 & 9.673 & 182.3 & 0.2398 & 3657 & 4.277 & 4.911 & 74.04 \\
RGG-6 & 55.1 & 7.112 & 40.69 & 226.1 & 0.2845 & 2212 & 2.484 & 12.48 & 42.91 \\
RGG-7 & 43.4 & 5.066 & 45.53 & 348.7 & 0.4789 & 16{,}850 & 19.70 & 14.08 & 39.22 \\
RGG-8 & 33.1 & 3.777 & 48.36 & 412.1 & 0.5889 & 56{,}810 & 68.60 & 11.60 & 22.28 \\
\bottomrule
\end{tabular}
\end{adjustbox}
\end{table}

The gaps are not monotone in graph size. RGG-6 has the largest
\texttt{RADAR} mean gap even though RGG-7 and RGG-8 are larger, while RGG-8 has the largest
gap in a single case. This is consistent with a local search whose quality depends on the
placement of routes that affect payoffs and on resource bottlenecks, not only on vertex count.
The maximum gaps also emphasize that \texttt{RADAR} is a heuristic with no general
approximation guarantee.

\subsection{Restricted Mixed Games versus Fully Enumerated Games}
\label{subsec:evaluation-rgr-full}

Table~\ref{tab:rgr-quality} compares restricted solutions from the combined generator with
$K_{\mathrm{cand}}=8$ against fully enumerated games. It reports $100|\delta_D|$. All 1,000
runs (500 Nash and 500 mixed SSE) in the five enumerable settings produce a restricted
solution and an exact defender payoff comparison.

\begin{table}[htbp]
\centering
\caption{\texttt{RGR} versus selected solutions of fully enumerated games using the combined
generator with $K_{\mathrm{cand}}=8$. Entries are $100|\delta_D|$ (\%).}
\label{tab:rgr-quality}
\begin{tabular}{lrrrr}
\toprule
& \multicolumn{2}{c}{Nash} & \multicolumn{2}{c}{Mixed SSE}\\
Topology & Mean & Max & Mean & Max \\
\midrule
Diamond-4 & 2.058 & 4.340 & 2.058 & 4.340 \\
Diamond-7 & 0 & 0 & 0 & 0 \\
RGG-6 & 0.040 & 0.819 & 0.041 & 0.819 \\
RGG-7 & 0.231 & 15.69 & 0.064 & 0.646 \\
RGG-8 & 0 & 0 & 0 & 0 \\
\midrule
All & 0.466 & 15.69 & 0.433 & 4.340 \\
\bottomrule
\end{tabular}
\end{table}

For Nash play, the mean absolute normalized difference is 0.466\%. Because the two games
may select different equilibria, this is a payoff difference rather than an optimality gap.
The 15.69\% maximum on RGG-7 is one outlier; the other 99 RGG-7 runs average 0.075\%. In the
outlying run, the route gap and measured gains from generated deviations are zero. Thus, no
generated action improves the selected restricted equilibrium, but an ungenerated defender
configuration or a different equilibrium of the fully enumerated game may still do so.

For mixed Stackelberg play, the mean absolute normalized difference is 0.433\% and the
maximum is 4.34\%. Diamond-4 has essentially exact route and generated action diagnostics
but a 2.058\% mean defender payoff difference across every candidate mode and count. Thus,
route response quality alone does not determine the value of a restricted game: an omitted
defender column, an omitted designated attacker response, or selection under strong
tie-breaking can still affect payoffs. Conversely, RGG-8 has zero defender payoff difference
even though the mean route gaps are $2.627\times10^{-4}$ for Nash and
$3.409\times10^{-4}$ for mixed Stackelberg play. The omitted routes in those cases do not
change the selected equilibrium payoff.

Increasing $K_{\mathrm{cand}}$ improves route response accuracy more consistently than
defender payoff agreement. A route gap evaluates one follower response at one defender
mixture, whereas the restricted game payoff also depends on defender columns, alternative
designated responses, tie-breaking, and equilibrium selection. A numerically accurate
current response is therefore not a certificate for the fully enumerated game.

\subsection{Larger Instances without Exhaustive References}
\label{subsec:evaluation-rgr-large}

The expanded records include larger settings for which complete action enumeration and
exact payoff comparison are unavailable. Table~\ref{tab:rgr-large} reports the combined
$K_{\mathrm{cand}}=8$ completion rate and mean generated action gap. These gaps cover only
generated actions.

\begin{table}[htbp]
\centering
\caption{\texttt{RGR} on larger settings with the combined generator and
$K_{\mathrm{cand}}=8$. No comparison with a fully enumerated game is available.}
\label{tab:rgr-large}
\begin{tabular}{llrr}
\toprule
Topology & Game & Completed (\%) & Mean generated action gap \\
\midrule
Diamond-10 & Nash & 100 & 0 \\
Diamond-10 & Mixed SSE & 100 & 0 \\
RGG-9 & Nash & 100 & $1.059\times10^{-3}$ \\
RGG-9 & Mixed SSE & 100 & $8.968\times10^{-4}$ \\
RGG-10 & Nash & 87 & $2.472\times10^{-1}$ \\
RGG-10 & Mixed SSE & 100 & $9.791\times10^{-5}$ \\
\bottomrule
\end{tabular}
\end{table}

Diamond-10 completes every run with zero generated action gaps. RGG-9 also completes every
run, although its remaining gaps are about $10^{-3}$. On RGG-10, the mixed SSE procedure
completes all 100 runs with a small mean gap, whereas the Nash procedure completes 87 and has
a much larger mean gap. These records show that the restricted method can remain operational
when exact enumeration is unavailable, but they do not establish proximity to a fully
enumerated solution. Here, ``completed'' means that the procedure returned under its
implemented stopping conditions, not that it certified the global game.

\subsection{Numerical Checks of Nash--Stackelberg Relations}
\label{subsec:evaluation-relations}

The evaluator includes one profile on three vertices for each of three payoff families:
heterogeneous single route payoffs, common ratio payoffs, and weighted zero-sum payoffs.
These are deterministic consistency checks, not a statistical experiment. Table~\ref{tab:relation-checks} summarizes the results.

\begin{table}[htbp]
\centering
\caption{Checks of the payoff relationships in Section~\ref{sec:relations}.}
\label{tab:relation-checks}
\begin{adjustbox}{max width=\textwidth}
\begin{tabular}{>{\raggedright\arraybackslash}p{0.28\textwidth}>{\raggedright\arraybackslash}p{0.48\textwidth}r}
\toprule
Profile & Checked statement & Result \\
\midrule
All three & Mixed SSE defender payoff is at least the mixed NE defender payoff & Pass; zero residual \\
Heterogeneous, single route & Mixed NE attacker payoff equals the minimax value & Pass; residual 0 \\
Heterogeneous, single route & Pure and mixed SSE attacker payoffs are at least the mixed NE payoff & Pass; zero residual \\
Common ratio & $U_D+\lambda U_A=C(a)$ and the two gap bounds hold & Pass; zero residual \\
Weighted zero sum & $U_D+\lambda U_A=0$, mixed SSE and mixed NE payoffs coincide, and pure/mixed ordering holds & Pass; zero residual \\
Weighted zero sum & Randomization gain identity \eqref{eq:randomization_gain} & Pass; zero residual \\
All three & A pure SSE automatically has no profitable unilateral defender deviation & Fail, with residuals $8$, $5$, and $8$ \\
\bottomrule
\end{tabular}
\end{adjustbox}
\end{table}

The failed final row is expected and useful. Proposition~\ref{prop:pure-se-is-ne} states a
condition for a pure Stackelberg outcome to be a pure NE; it does not claim that the
condition holds automatically. The positive residuals are direct counterexamples to that
false implication. All identities and inequalities that are actually predicted by
Theorems~\ref{thm:mixed-sse-defender-dominates-ne}--\ref{thm:weighted-zero-sum-relations} pass on their
corresponding profiles.

\subsection{Evaluation Summary and Limitations}
\label{subsec:evaluation-summary}

The experiments support the structural divide developed in the analysis. Against an
observed pure configuration, routing is exact and inexpensive, so \texttt{RADAR} can test
hundreds of local configurations while remaining under 1~ms on average. Against a mixed
defender strategy, exact routing and defender response are NP-hard; the combined route
generator is highly accurate on the enumerable cases, and the restricted games often match
selected fully enumerated payoffs closely, but gaps over generated actions are not global
guarantees.
The RGG-7 Nash outlier and the larger RGG-10 Nash results show why this distinction matters.

The numerical conclusions have four principal limitations. Exact comparisons stop at
Diamond-7 and RGG-8. Normalized gaps depend on the enumerated payoff range $R_D$. Nash
payoff differences can reflect equilibrium selection rather than unilateral regret. Finally,
the supplied records omit several generation parameters listed in
Section~\ref{subsec:evaluation-setup}. Appendix~\ref{app:detailed-evaluation} reports the
aggregates underlying the condensed tables.

\section{Conclusion}
\label{sec:conclusion}

We studied resource-aware intrusion detection in infrastructure networks where sensing type,
processing mode, deployment cost, edge server capacity, false alarms, and attacks that adapt
their routes interact in one graph model. The same pure action spaces support three strategic
settings. Simultaneous play may have no pure Nash equilibrium but always has a mixed
one. When the attacker observes a pure defender configuration, its exact best response
reduces to one shortest path computation followed by independent decisions about which
targets to attack. The defender's pure Stackelberg optimization remains NP-hard. When the
attacker observes a mixed defender strategy but not its realization, a mixed SSE exists, a
pessimistic optimum may fail to be attained, and both exact response oracles and mixed
Stackelberg optimization are NP-hard.

These results motivate different computational roles for the proposed methods.
\texttt{RADAR} exploits the inexpensive exact response to evaluate repairs, upgrades, and
exchanges for pure configurations. \texttt{RGR} uses the same configuration search within
restricted Nash and mixed Stackelberg games, together with support and geometric surrogate
route candidates. On enumerable instances, the methods have small mean normalized
differences from exact or fully enumerated references. Larger experiments show that
restricted computation remains possible when exhaustive enumeration is unavailable. The
expanded proofs and diagnostics also clarify the limitations: local search has no general
approximation bound, and generated action gaps are not global certificates when the
underlying response problems are NP-hard.

The payoff analysis separates the effects of observability from those of target utilities.
Under strong tie-breaking, a mixed SSE always gives the defender at least its mixed Nash
payoff.
Attacks on a single route give the attacker a minimax characterization even with
heterogeneous target values. A common ratio between defender and attacker payoff differences
extends that characterization to coordinated attacks on multiple routes and bounds both
payoff differences. In the weighted zero-sum case, mixed Nash and mixed Stackelberg payoffs
coincide, and the value of defender randomization equals the corresponding reduction in
attacker payoff.

\appendix

\section{Proof of Theorem~\ref{thm:defender_pure_se_hard}}
\label{app:defender-nph-proof}

\begin{proof}
We give a reduction, computable in polynomial time, from the decision version of \textsc{Clique}
\cite{karp1972reducibility}. An instance consists of an undirected graph
$H=(V_H,E_H)$ and an integer $r$. The question is whether $H$ contains a set of $r$
vertices with every pair adjacent. Let $n=|V_H|$, assume $1\le r\le n$, and let $d_H(v)$
denote the degree of vertex $v$ in $H$.

\paragraph{Construction.}
Set $M=2n+1$. Build a monitored graph $G_H$ with one source vertex $s$ and one target
vertex $t_v$ for every $v\in V_H$. For each input edge $\{u,v\}\in E_H$, add the edge
$\{t_u,t_v\}$ between the corresponding target vertices. For every $v\in V_H$, add
$M-d_H(v)$ private intermediate vertices
\[
 x_{v,1},\ldots,x_{v,M-d_H(v)}.
\]
For each private vertex $x_{v,j}$, add edges $\{s,x_{v,j}\}$ and
$\{x_{v,j},t_v\}$. Thus each private vertex forms a separate path of length two from $s$ to
$t_v$. There are no other vertices or edges.

Every edge supports one perfect sensing action of unit cost, zero false alarm probability,
and zero server load. The only active feasibility constraint is a cardinality budget
\begin{equation}
 B_D=rM-r(r-1).
 \label{eq:clique-budget}
\end{equation}
Every target has payoffs
\[
 U_A^u=1,\qquad U_A^d=-1,
 \qquad
 U_D^u=-1,\qquad U_D^d=1.
\]
The construction contains
$1+n+\sum_v(M-d_H(v))=O(n^2)$ vertices and $O(n^2+|E_H|)$ edges, and all numerical
parameters require polynomial encoding length.

Call a target $t_v$ \emph{fully covered} if every route from $s$ to $t_v$ contains at
least one selected sensor. Because the sensors are perfect, a fully covered target has
maximum route non-detection probability zero and best route attacker utility $-1$, so the
attacker omits it. If a target is not fully covered, a route with no sensed edge remains
undetected with probability one and gives attacker utility $1$, so every attacker best
response includes an attack on that target. The target decisions are independent and strict,
so there are no attacker ties that affect payoffs. If exactly $\tau$ targets are fully covered,
the attacker attacks the other $n-\tau$ targets, each of which gives defender utility $-1$.
The defender value is
therefore
\begin{equation}
 -(n-\tau).
 \label{eq:clique-defender-value}
\end{equation}
Maximizing defender utility is equivalent to maximizing the number of fully covered targets.

\paragraph{A cut lower bound.}
For any vertex set $U\subseteq V(G_H)$ with $s\notin U$, define
\[
 W(U)=\{v\in V_H:t_v\in U\}.
\]
Let $E_H(W)$ be the set of input edges with both endpoints in $W$, and let
$\delta_H(W)$ be the set of input edges with exactly one endpoint in $W$. Let
$\delta_{G_H}(U)$ be the cut in the constructed graph.

Fix $v\in W(U)$. For each private path
$s-x_{v,j}-t_v$, the source $s$ lies outside $U$ and the target $t_v$ lies inside $U$.
Whether $x_{v,j}$ lies inside or outside $U$, at least one of the two path edges crosses
$\delta_{G_H}(U)$. These crossing edges are distinct across private paths. Hence the
private paths contribute at least
$\sum_{v\in W(U)}(M-d_H(v))$ cut edges. Every target--target edge corresponding to an input
edge in $\delta_H(W(U))$ also crosses the constructed cut. Additional cut edges can only
increase its size, so
\begin{align}
 |\delta_{G_H}(U)|
 &\ge \sum_{v\in W(U)}(M-d_H(v))
      +|\delta_H(W(U))| \\
 &=M|W(U)|-\sum_{v\in W(U)}d_H(v)
      +|\delta_H(W(U))|.
 \label{eq:clique-cut-first}
\end{align}
The degree identity
\[
 \sum_{v\in W}d_H(v)=2|E_H(W)|+|\delta_H(W)|
\]
then gives
\begin{equation}
 |\delta_{G_H}(U)|
 \ge M|W(U)|-2|E_H(W(U))|.
 \label{eq:clique_cut_bound}
\end{equation}

\paragraph{If $H$ has an $r$-clique, the defender covers $r$ targets.}
Suppose $C\subseteq V_H$ is an $r$-clique. Let
$U_C=\{t_v:v\in C\}$ and select one sensor on every edge of
$\delta_{G_H}(U_C)$. Every source--target route to a target in $U_C$ must cross this cut,
so all $r$ corresponding targets are fully covered. The cut contains one edge from each
private path incident to a target in $C$, together with the target--target edges leaving
$C$. Its size is
\begin{align*}
 |\delta_{G_H}(U_C)|
 &=\sum_{v\in C}(M-d_H(v))+|\delta_H(C)|\\
 &=rM-2|E_H(C)|\\
 &=rM-2\binom r2\\
 &=rM-r(r-1)=B_D.
\end{align*}
Thus the selected sensors are feasible and give defender utility at least $-(n-r)$ by
\eqref{eq:clique-defender-value}.

\paragraph{If the defender covers at least $r$ targets, $H$ has an $r$-clique.}
Conversely, suppose a set $\mathcal Y$ of at most $B_D$ selected edges fully covers at least $r$
targets. Remove the edges in $\mathcal Y$ from $G_H$, and let $U$ be the union of all connected
components that do not contain $s$. Every fully covered target is disconnected from $s$ and
therefore lies in $U$. Hence
\[
 \tau=|W(U)|\ge r.
\]
Every edge crossing $\delta_{G_H}(U)$ must have been removed; otherwise it would connect a
component in $U$ to the component containing $s$. Consequently,
\begin{equation}
 |\delta_{G_H}(U)|\le |\mathcal Y|\le B_D.
 \label{eq:clique-cut-upper}
\end{equation}
A set of $\tau$ vertices induces at most $\binom{\tau}{2}$ edges, so
\eqref{eq:clique_cut_bound} gives
\begin{equation}
 |\delta_{G_H}(U)|
 \ge \tau M-\tau(\tau-1).
 \label{eq:clique-q-bound}
\end{equation}
Define $f(x)=xM-x(x-1)$. For every integer $x<n$,
\[
 f(x+1)-f(x)=M-2x=2n+1-2x>0.
\]
Thus $f$ is strictly increasing on $\{1,\ldots,n\}$. If $\tau>r$, then
\eqref{eq:clique-q-bound} implies
$|\delta_{G_H}(U)|\ge f(\tau)>f(r)=B_D$, contradicting
\eqref{eq:clique-cut-upper}. Therefore $\tau=r$.

Substituting $\tau=|W(U)|=r$ into \eqref{eq:clique_cut_bound} and using
\eqref{eq:clique-cut-upper} yields
\begin{align*}
 rM-2|E_H(W(U))|
 &\le |\delta_{G_H}(U)|\\
 &\le rM-r(r-1).
\end{align*}
Hence
\[
 |E_H(W(U))|\ge\frac{r(r-1)}2=\binom r2.
\]
No $r$-vertex set can induce more than $\binom r2$ edges, so equality holds and every pair
of vertices in $W(U)$ is adjacent. Thus $W(U)$ is an $r$-clique.

We have shown that $H$ contains an $r$-clique if and only if the defender can attain
utility at least $-(n-r)$. Therefore deciding whether the optimal pure Stackelberg value
reaches that threshold is NP-hard. The attacker decisions are strict in the construction,
so strong and pessimistic tie-breaking give the same value. This proves the theorem under
both rules.
\end{proof}

\section{Proof of Theorem~\ref{thm:mixed_existence}}
\label{app:mixed-existence-proof}

\begin{proof}
We first prove existence under strong tie-breaking. For each attacker action
$a\in\mathcal A$, define its best response region in the defender simplex by
\begin{equation}
 P_a=
 \left\{\sigma\in\Delta(\mathcal X):
 \bar U_A(\sigma,a)\ge\bar U_A(\sigma,a'),
 \ \forall a'\in\mathcal A\right\}.
 \label{eq:best-response-polytope}
\end{equation}
Each inequality is linear and weak, so $P_a$ is a closed polytope. It is contained in the
compact simplex $\Delta(\mathcal X)$ and is therefore compact. Some $P_a$ may be empty, but
at least one is nonempty: for every defender mixture, the finite attacker action set has a
best response, so the union of the $P_a$ covers the simplex.

For every $a$ with $P_a\ne\emptyset$, the defender payoff
$\bar U_D(\sigma,a)$ is linear and continuous in $\sigma$. It therefore attains a maximum
on compact $P_a$. Choose
\begin{equation}
 \sigma_a\in\argmax_{\sigma\in P_a}\bar U_D(\sigma,a).
 \label{eq:best-mixture-for-response}
\end{equation}
There are finitely many attacker actions, so among the nonempty regions choose
\begin{equation}
 a^\star\in\argmax_{a:P_a\ne\emptyset}
 \bar U_D(\sigma_a,a).
 \label{eq:best-response-region}
\end{equation}
Because $\sigma_{a^\star}\in P_{a^\star}$, the action $a^\star$ is an attacker best
response to $\sigma_{a^\star}$.

Now take any defender mixture $\sigma$ and any
$a\in\BR_A(\sigma)$. Then $\sigma\in P_a$, so by
\eqref{eq:best-mixture-for-response},
\[
 \bar U_D(\sigma,a)
 \le \bar U_D(\sigma_a,a)
 \le \bar U_D(\sigma_{a^\star},a^\star).
\]
Thus no feasible pair consisting of a defender mixture and one of its attacker best
responses gives the defender more than
$(\sigma_{a^\star},a^\star)$. This pair maximizes the strong Stackelberg objective and is a
mixed SSE\@.

We next construct an instance in which the pessimistic mixed problem does not attain its
supremum. There is one target $t$ with one route $r$. The defender has two feasible
configurations. Configuration $z_0$ leaves the route unmonitored, so $q_r(z_0)=1$.
Configuration $z_1$ detects with certainty, so $q_r(z_1)=0$. Set
\[
 U_D^u(t)=-1,
 \quad U_D^d(t)=3,
 \quad U_A^u(t)=1,
 \quad U_A^d(t)=-1.
\]
The attacker may use $a_r=\{r\}$ or the empty action. The pure payoff table is
\begin{equation*}
\begin{array}{c|cc}
 & a_r & \emptyset\\ \hline
 z_0 & (-1,1) & (0,0)\\
 z_1 & (3,-1) & (0,0).
\end{array}
\end{equation*}

Let $p=\sigma(z_1)$, so $\sigma(z_0)=1-p$. If the attacker uses the route, its expected
utility is
\begin{align*}
 \bar U_A(\sigma,a_r)
 &=(1-p)(1)+p(-1)=1-2p,
\end{align*}
and the defender's expected utility is
\begin{align*}
 \bar U_D(\sigma,a_r)
 &=(1-p)(-1)+p(3)=-1+4p.
\end{align*}
The empty action gives both players zero for every $p$.

If $p<1/2$, then $1-2p>0$, so $a_r$ is the unique attacker best response and the
pessimistic defender value is $-1+4p$. If $p>1/2$, then $1-2p<0$, so the empty action is
the unique best response and the defender value is zero. At $p=1/2$, both attacker actions
are best responses. They give the defender utilities $1$ and $0$, respectively, so
pessimistic tie-breaking selects the empty action and gives value zero. Therefore
\begin{equation}
 V_D^-(p)=
 \begin{cases}
 -1+4p, & 0\le p<1/2,\\
 0, & 1/2\le p\le1.
 \end{cases}
 \label{eq:pessimistic-nonattainment-value}
\end{equation}
As $p$ approaches $1/2$ from below, $V_D^-(p)$ approaches $1$. No $p$ attains value $1$:
all $p<1/2$ give strictly less than $1$, and all $p\ge1/2$ give zero. Hence the pessimistic
supremum is $1$ but is not attained. This completes both parts of the theorem.
\end{proof}

\section{Proof of Theorem~\ref{thm:mixed_oracle_hard}}
\label{app:mixed-oracle-hard-proof}

\begin{proof}
We prove the three claims separately. All reductions use graphs of polynomial size and
rational parameters.

\paragraph{Part (i): attacker best response via \textsc{Max-Cut}.}
An instance of \textsc{Max-Cut} is an undirected graph $H=(V_H,E_H)$. A cut assigns each
vertex a binary label and counts the edges whose endpoints receive different labels.
Computing a maximum cut is NP-hard \cite{karp1972reducibility}. Let
$V_H=\{1,\ldots,n\}$ and $m_H=|E_H|$. We may assume $m_H\ge1$, since an edgeless instance is
trivial.

Construct a chain of $n$ diamond gadgets. Introduce articulation vertices
$x_0,x_1,\ldots,x_n$, with source $s=x_0$ and the single target $t=x_n$. For every input
vertex $i$ and bit $b\in\{0,1\}$, introduce a private branch vertex $y_i^b$ and edges
\[
 \{x_{i-1},y_i^b\},\qquad \{y_i^b,x_i\}.
\]
The two branches of length two form the $i$th diamond. Every simple $s$--$t$ route must pass
through the articulation vertices in order and choose exactly one branch in every diamond.
Thus each bit vector $\mathbf b=(b_1,\ldots,b_n)\in\{0,1\}^n$ uniquely determines a route $r_{\mathbf b}$,
and every such route determines one bit vector.

Designate the first edge on each branch,
$g_i^b=\{x_{i-1},y_i^b\}$, as an edge that supports sensing. It has one perfect sensing
action with unit cost, zero false alarm probability, and zero server load. Set the defender
budget to two, so every configuration used below is feasible.

For each input edge $\{i,j\}\in E_H$, include two pure configurations in the support of the
mixed defender strategy. Configuration $z_{ij}^{01}$ perfectly senses gates $g_i^1$ and
$g_j^0$ and leaves every other gate unsensed. Route $r_{\mathbf b}$ remains undetected under this
configuration if and only if $(b_i,b_j)=(0,1)$. Configuration $z_{ij}^{10}$ senses
$g_i^0$ and $g_j^1$, so $r_{\mathbf b}$ remains undetected if and only if
$(b_i,b_j)=(1,0)$. Let $\sigma$ be uniform over these $2m_H$ explicitly listed
configurations.

For route $r_{\mathbf b}$, the non-detection probability under each pure configuration is either zero
or one. An input edge with $b_i=b_j$ contributes no orientation under which the route remains
undetected. An input edge with $b_i\ne b_j$ contributes exactly one of its two orientations.
Therefore
\begin{equation}
 \bar q_{r_{\mathbf b}}(\sigma)
 =\frac{|\{\{i,j\}\in E_H:b_i\ne b_j\}|}{2m_H}.
 \label{eq:max-cut-survival}
\end{equation}
Set the single target's attacker payoffs to
\[
 U_A^u=2m_H,
 \qquad
 U_A^d=-1.
\]
The expected utility of route $r_{\mathbf b}$ is
\begin{equation}
 \bar U_A(\sigma,r_{\mathbf b})
 =-1+(2m_H+1)\bar q_{r_{\mathbf b}}(\sigma),
 \label{eq:max-cut-utility}
\end{equation}
which is strictly increasing in the cut size. Because $m_H\ge1$, some cut contains at least
one edge, and its route utility is
$-1+(2m_H+1)/(2m_H)=1/(2m_H)>0$. Thus an optimal route strictly dominates the empty action.
Decoding the branch choice of an exact attacker best response yields a maximum cut of $H$.
The graph has $O(n)$ vertices and edges, and the mixed strategy has support $2m_H$, so the
reduction is polynomial. Computing an exact attacker best response is NP-hard even with one
target and hence at most one attacked route.

\paragraph{Part (ii): pure defender best response via \textsc{Maximum Coverage}.}
An instance of \textsc{Maximum Coverage} consists of a universe
$\Omega=\{1,\ldots,N\}$, subsets $S_1,\ldots,S_L\subseteq\Omega$, and a budget
$B_C$. The objective is to select at most $B_C$ subsets whose union has maximum cardinality.
The problem is NP-hard because deciding whether $B_C$ selected sets can cover all elements is
the NP-complete \textsc{Set Cover} decision problem \cite{karp1972reducibility}.

Create common layer vertices $v_0,v_1,\ldots,v_L$. For each set $S_j$, add a direct gate
edge that supports sensing,
\[
 e_j=\{v_{j-1},v_j\}.
\]
For every element $u\notin S_j$, add a private bypass vertex $b_{u,j}$ and edges
$\{v_{j-1},b_{u,j}\}$ and $\{b_{u,j},v_j\}$. For every element $u\in\Omega$, add a
private source $s_u$ connected to $v_0$ and a private target $t_u$ connected to $v_L$.
Define a prescribed route $r_u$ from $s_u$ to $t_u$ as follows: in layer $j$, it traverses
the gate edge $e_j$ if $u\in S_j$ and traverses its private bypass if $u\notin S_j$.
The unique source and target edges make the prescribed routes distinct even if two elements
have the same membership pattern.

Only the gate edges $e_j$ support sensing. Each has one perfect sensing action of unit cost,
and the defender budget is $B_C$. Set every target's defender payoffs to
$U_D^u=-1$ and $U_D^d=1$. Let the mixed attacker strategy $\mu$ be uniform over the
explicitly listed singleton actions $\{r_u\}$, one for each element.

Selecting sensors on a set of gates indexed by $J\subseteq\{1,\ldots,L\}$ detects route
$r_u$ with certainty if and only if $u\in\bigcup_{j\in J}S_j$. A covered element therefore
contributes defender utility $1$, while an uncovered element contributes $-1$. If $C(J)$ is
the number of covered elements, the expected defender utility against $\mu$ is
\begin{equation}
 \bar U_D(J,\mu)
 =\frac{C(J)-(N-C(J))}{N}
 =\frac{2C(J)}N-1.
 \label{eq:max-coverage-defender-utility}
\end{equation}
This expression is strictly increasing in $C(J)$, and the defender feasibility condition is
exactly $|J|\le B_C$. Hence any pure defender best response identifies a maximum coverage
selection. The graph and the explicit attacker support are polynomial in the input size,
proving NP-hardness.

\paragraph{Part (iii): optimal mixed SSE via $0$--$1$ \textsc{Knapsack}.}
An instance of $0$--$1$ \textsc{Knapsack} has items $i=1,\ldots,n$ with positive integer
costs $c_i$, positive values $v_i$, and budget $B$. The objective is to choose a subset of
total cost at most $B$ and maximum total value; this problem is NP-hard
\cite{karp1972reducibility}.

Construct a star with common source $s$, one target $t_i$ per item, and one edge
$e_i=\{s,t_i\}$. Edge $e_i$ supports a single sensing action with cost $c_i$, detection
probability $1/2$, zero false alarm probability, and zero server load. The defender budget
is $B$. The unique route to target $t_i$ is $r_i=e_i$. Set
\begin{equation}
 U_A^u(t_i)=1,
 \qquad U_A^d(t_i)=-\frac14,
 \qquad
 U_D^u(t_i)=-v_i,
 \qquad U_D^d(t_i)=v_i.
 \label{eq:knapsack-payoffs}
\end{equation}

If $e_i$ is unsensed, route $r_i$ remains undetected with probability one and gives attacker
utility $1$. If $e_i$ is sensed, it remains undetected with probability $1/2$ and gives
attacker utility
\[
 \frac12(1)+\frac12\left(-\frac14\right)=\frac38>0.
\]
For any mixed defender strategy, the expected route utility lies between $3/8$ and $1$ and
is therefore strictly positive. Since there is one route per target and utilities are
additive, the unique attacker best response to every pure or mixed defender strategy is the
action containing all routes $\{r_1,\ldots,r_n\}$. Tie-breaking plays no role.

For the defender, an unsensed item route gives utility $-v_i$, while a sensed item route
gives
\[
 \frac12(-v_i)+\frac12(v_i)=0.
\]
Thus sensing item edge $e_i$ raises defender utility by exactly $v_i$. For a pure feasible
subset $J$, the defender value is
\begin{equation}
 -\sum_{i=1}^n v_i+\sum_{i\in J}v_i,
 \qquad
 \sum_{i\in J}c_i\le B.
 \label{eq:knapsack-defender-value}
\end{equation}
Maximizing \eqref{eq:knapsack-defender-value} is exactly the knapsack optimization problem,
up to the additive constant $-\sum_i v_i$.

Finally, randomization cannot improve this value. Every mixed defender strategy is a
distribution over feasible pure subsets, and the attacker response is the same for every
realization and every distribution. Its defender payoff is therefore a convex combination
of the pure values in \eqref{eq:knapsack-defender-value}, which cannot exceed their maximum.
Hence an optimal mixed SSE has the optimal knapsack value and may be chosen pure. Computing
an optimal mixed SSE is NP-hard.

The three reductions establish parts (i)--(iii).
\end{proof}

\section{Proof of Lemma~\ref{lem:single_route-bound}}
\label{app:single-route-proof}

\begin{proof}
Let $(\sigma^{\mathrm{NE}},\mu^{\mathrm{NE}})$ be a mixed Nash equilibrium, and let
\[
 v_A=\bar U_A(\sigma^{\mathrm{NE}},\mu^{\mathrm{NE}})
\]
be the attacker equilibrium payoff. We consider separately whether the empty action belongs
to the support of $\mu^{\mathrm{NE}}$.

\paragraph{Case 1: the empty action is not in the attacker support.}
Every supported attacker action is then a singleton route. Write
$\mu^{\mathrm{NE}}(r)$ for the probability assigned to action $\{r\}$. Define
\begin{equation}
 \kappa=
 \sum_{r\in\mathcal R}
 \mu^{\mathrm{NE}}(r)
 \frac{\Delta_D(t(r))}{\Delta_A(t(r))}.
 \label{eq:single-route-kappa}
\end{equation}
Every term is nonnegative, at least one route has positive probability, and both payoff
differences are positive, so $\kappa>0$. Define a reweighted route mixture
\begin{equation}
 \widehat\mu(r)=
 \frac{
 \mu^{\mathrm{NE}}(r)\Delta_D(t(r))/\Delta_A(t(r))
 }{\kappa}.
 \label{eq:attacker_reweighting}
\end{equation}
The weights are nonnegative and sum to one, so $\widehat\mu$ is a valid attacker mixture
with support contained in $\supp(\mu^{\mathrm{NE}})$.

For a defender mixture $\sigma$, let
$\bar q_r(\sigma)=\sum_z\sigma(z)q_r(z)$. From
\eqref{eq:defender_route_utility} and \eqref{eq:attacker_route_utility}, the change in
defender payoff against $\mu^{\mathrm{NE}}$ between any two defender mixtures
$\sigma_1$ and $\sigma_2$ is
\begin{align}
 &\bar U_D(\sigma_1,\mu^{\mathrm{NE}})
 -\bar U_D(\sigma_2,\mu^{\mathrm{NE}})\nonumber\\
 &\quad=-\sum_r\mu^{\mathrm{NE}}(r)\Delta_D(t(r))
 \bigl[\bar q_r(\sigma_1)-\bar q_r(\sigma_2)\bigr].
 \label{eq:single-route-defender-difference}
\end{align}
The corresponding change in attacker payoff against $\widehat\mu$ is
\begin{align}
 &\bar U_A(\sigma_1,\widehat\mu)
 -\bar U_A(\sigma_2,\widehat\mu)\nonumber\\
 &\quad=\sum_r\widehat\mu(r)\Delta_A(t(r))
 \bigl[\bar q_r(\sigma_1)-\bar q_r(\sigma_2)\bigr]\nonumber\\
 &\quad=\frac1\kappa
 \sum_r\mu^{\mathrm{NE}}(r)\Delta_D(t(r))
 \bigl[\bar q_r(\sigma_1)-\bar q_r(\sigma_2)\bigr].
 \label{eq:single-route-attacker-difference}
\end{align}
Combining the two equations gives
\begin{equation}
 \bar U_D(\sigma_1,\mu^{\mathrm{NE}})
 -\bar U_D(\sigma_2,\mu^{\mathrm{NE}})
 =-\kappa\left[
 \bar U_A(\sigma_1,\widehat\mu)
 -\bar U_A(\sigma_2,\widehat\mu)
 \right].
 \label{eq:single-route-difference-identity}
\end{equation}

The defender Nash condition says that $\sigma^{\mathrm{NE}}$ maximizes defender utility
against $\mu^{\mathrm{NE}}$. Applying
\eqref{eq:single-route-difference-identity} with
$\sigma_1=\sigma^{\mathrm{NE}}$ and arbitrary $\sigma_2=\sigma$ shows that
\begin{equation}
 \bar U_A(\sigma^{\mathrm{NE}},\widehat\mu)
 \le \bar U_A(\sigma,\widehat\mu),
 \qquad \forall\sigma\in\Delta(\mathcal X).
 \label{eq:single-route-defender-minimizes}
\end{equation}
Thus $\sigma^{\mathrm{NE}}$ minimizes attacker payoff against $\widehat\mu$.

Every route in $\supp(\mu^{\mathrm{NE}})$ is an attacker best response to
$\sigma^{\mathrm{NE}}$ and gives the common payoff $v_A$. The reweighted mixture has the
same support subset, so it is also a best response and
\begin{equation}
 \bar U_A(\sigma^{\mathrm{NE}},\widehat\mu)=v_A.
 \label{eq:single-route-reweighted-value}
\end{equation}
Equations \eqref{eq:single-route-defender-minimizes} and
\eqref{eq:single-route-reweighted-value}, together with the fact that
$\widehat\mu$ maximizes attacker payoff against $\sigma^{\mathrm{NE}}$, make
$(\sigma^{\mathrm{NE}},\widehat\mu)$ a saddle point of the zero-sum game whose payoff to the
attacker is $\bar U_A$. Therefore
\begin{equation}
 v_A=
 \min_{\sigma'\in\Delta(\mathcal X)}
 \max_{a'\in\mathcal A}\bar U_A(\sigma',a').
 \label{eq:single-route-minimax-equality}
\end{equation}
For any defender mixture $\sigma$ and any best response $a\in\BR_A(\sigma)$,
\[
 \bar U_A(\sigma,a)
 =\max_{a'\in\mathcal A}\bar U_A(\sigma,a')
 \ge\min_{\sigma'}\max_{a'}\bar U_A(\sigma',a')
 =v_A.
\]
This proves the lemma in the first case.

\paragraph{Case 2: the empty action has positive Nash probability.}
Every pure action in the support of a mixed best response gives the same payoff. The empty
action always gives zero, so $v_A=0$. Since it is a best response to
$\sigma^{\mathrm{NE}}$, no attacker action gives more than zero at that defender mixture;
hence
\[
 \max_{a\in\mathcal A}\bar U_A(\sigma^{\mathrm{NE}},a)=0.
\]
For every defender mixture $\sigma$, the attacker can choose the empty action and obtain
zero, so
\[
 \max_{a\in\mathcal A}\bar U_A(\sigma,a)\ge0.
\]
Consequently,
\[
 \min_{\sigma}\max_a\bar U_A(\sigma,a)=0=v_A.
\]
Any attacker best response to any $\sigma$ attains the corresponding maximum and therefore
has payoff at least zero. This again yields \eqref{eq:single_route_bound} and completes the
proof.
\end{proof}

\clearpage
\begin{landscape}
\section{Detailed Evaluation Tables}
\label{app:detailed-evaluation}

This appendix reports the aggregates underlying Section~\ref{sec:evaluation}.
The tables retain the evaluator grouping by topology, game, candidate source, candidate
count, and tie-breaking rule. Values marked ``--'' were not available because exhaustive
enumeration was not performed. The candidate path cap was never reached, so that redundant
column is omitted. As in Section~\ref{sec:evaluation}, every reported value with magnitude
below $10^{-12}$ is displayed as $0$.

\scriptsize
\setlength{\tabcolsep}{3pt}
\begin{longtable}{@{}lrrrr@{}}
\caption{Exact checks for the Dijkstra response to an observed pure configuration.}\label{tab:detail-fixed-response}\\
\toprule
Topology & Runs & Exact checks & Within tolerance (\%) & Search \\
\midrule
\endfirsthead
\multicolumn{5}{l}{\tablename~\thetable\ continued from the previous page}\\
\toprule
Topology & Runs & Exact checks & Within tolerance (\%) & Search \\
\midrule
\endhead
\midrule
\multicolumn{5}{r}{Continued on the next page}\\
\endfoot
\bottomrule
\endlastfoot
Diamond-4 & 100 & 100/100 & 100.0 & Dijkstra \\
Diamond-7 & 100 & 100/100 & 100.0 & Dijkstra \\
RGG-6 & 100 & 100/100 & 100.0 & Dijkstra \\
RGG-7 & 100 & 100/100 & 100.0 & Dijkstra \\
RGG-8 & 100 & 100/100 & 100.0 & Dijkstra \\
\end{longtable}
\end{landscape}

\begin{landscape}
\scriptsize
\setlength{\tabcolsep}{3pt}
\begin{longtable}{@{}lllrrrrr@{}}
\caption{Route generation against mixed defender strategies on instances with exhaustive
attacker comparisons.}\label{tab:detail-mixed-response}\\
\toprule
Topology & Game & Candidate source & $K_{\mathrm{cand}}$ & Runs & Exact (\%) & Mean route gap & Max route gap \\
\midrule
\endfirsthead
\multicolumn{8}{l}{\tablename~\thetable\ continued from the previous page}\\
\toprule
Topology & Game & Candidate source & $K_{\mathrm{cand}}$ & Runs & Exact (\%) & Mean route gap & Max route gap \\
\midrule
\endhead
\midrule
\multicolumn{8}{r}{Continued on the next page}\\
\endfoot
\bottomrule
\endlastfoot
Diamond-4 & Nash & Support & -- & 300 & 100.0 & 0 & 0 \\
Diamond-4 & Nash & Surrogate & 1 & 100 & 100.0 & 0 & 0 \\
Diamond-4 & Nash & Surrogate & 4 & 100 & 100.0 & 0 & 0 \\
Diamond-4 & Nash & Surrogate & 8 & 100 & 100.0 & 0 & 0 \\
Diamond-4 & Nash & Combined & 1 & 100 & 100.0 & 0 & 0 \\
Diamond-4 & Nash & Combined & 4 & 100 & 100.0 & 0 & 0 \\
Diamond-4 & Nash & Combined & 8 & 100 & 100.0 & 0 & 0 \\
Diamond-4 & Mixed SSE & Support & -- & 300 & 100.0 & 0 & 0 \\
Diamond-4 & Mixed SSE & Surrogate & 1 & 100 & 100.0 & 0 & 0 \\
Diamond-4 & Mixed SSE & Surrogate & 4 & 100 & 100.0 & 0 & 0 \\
Diamond-4 & Mixed SSE & Surrogate & 8 & 100 & 100.0 & 0 & 0 \\
Diamond-4 & Mixed SSE & Combined & 1 & 100 & 100.0 & 0 & 0 \\
Diamond-4 & Mixed SSE & Combined & 4 & 100 & 100.0 & 0 & 0 \\
Diamond-4 & Mixed SSE & Combined & 8 & 100 & 100.0 & 0 & 0 \\
Diamond-7 & Nash & Support & -- & 300 & 26.0 & 0.01196 & 0.08003 \\
Diamond-7 & Nash & Surrogate & 1 & 100 & 99.0 & 0.0001277 & 0.01277 \\
Diamond-7 & Nash & Surrogate & 4 & 100 & 100.0 & 0 & 0 \\
Diamond-7 & Nash & Surrogate & 8 & 100 & 100.0 & 0 & 0 \\
Diamond-7 & Nash & Combined & 1 & 100 & 100.0 & 0 & 0 \\
Diamond-7 & Nash & Combined & 4 & 100 & 100.0 & 0 & 0 \\
Diamond-7 & Nash & Combined & 8 & 100 & 100.0 & 0 & 0 \\
Diamond-7 & Mixed SSE & Support & -- & 300 & 39.0 & 0.009809 & 0.08058 \\
Diamond-7 & Mixed SSE & Surrogate & 1 & 100 & 98.0 & 0.0003488 & 0.02818 \\
Diamond-7 & Mixed SSE & Surrogate & 4 & 100 & 100.0 & 0 & 0 \\
Diamond-7 & Mixed SSE & Surrogate & 8 & 100 & 100.0 & 0 & 0 \\
Diamond-7 & Mixed SSE & Combined & 1 & 100 & 100.0 & 0 & 0 \\
Diamond-7 & Mixed SSE & Combined & 4 & 100 & 100.0 & 0 & 0 \\
Diamond-7 & Mixed SSE & Combined & 8 & 100 & 100.0 & 0 & 0 \\
RGG-6 & Nash & Support & -- & 300 & 52.0 & 0.01792 & 0.1297 \\
RGG-6 & Nash & Surrogate & 1 & 100 & 94.0 & 0.0008331 & 0.01468 \\
RGG-6 & Nash & Surrogate & 4 & 100 & 99.0 & 0.0001421 & 0.01421 \\
RGG-6 & Nash & Surrogate & 8 & 100 & 100.0 & 0 & 0 \\
RGG-6 & Nash & Combined & 1 & 100 & 94.0 & 0.000768 & 0.02206 \\
RGG-6 & Nash & Combined & 4 & 100 & 100.0 & 0 & 0 \\
RGG-6 & Nash & Combined & 8 & 100 & 100.0 & 0 & 0 \\
RGG-6 & Mixed SSE & Support & -- & 300 & 47.0 & 0.01724 & 0.119 \\
RGG-6 & Mixed SSE & Surrogate & 1 & 100 & 93.0 & 0.0008019 & 0.01468 \\
RGG-6 & Mixed SSE & Surrogate & 4 & 100 & 100.0 & 0 & 0 \\
RGG-6 & Mixed SSE & Surrogate & 8 & 100 & 100.0 & 0 & 0 \\
RGG-6 & Mixed SSE & Combined & 1 & 100 & 95.0 & 0.0005901 & 0.0142 \\
RGG-6 & Mixed SSE & Combined & 4 & 100 & 100.0 & 0 & 0 \\
RGG-6 & Mixed SSE & Combined & 8 & 100 & 100.0 & 0 & 0 \\
RGG-7 & Nash & Support & -- & 300 & 47.0 & 0.02436 & 0.1056 \\
RGG-7 & Nash & Surrogate & 1 & 100 & 88.0 & 0.00144 & 0.04533 \\
RGG-7 & Nash & Surrogate & 4 & 100 & 93.0 & 0.000193 & 0.007716 \\
RGG-7 & Nash & Surrogate & 8 & 100 & 99.0 & $2.756\times10^{-6}$ & 0.0002756 \\
RGG-7 & Nash & Combined & 1 & 100 & 90.0 & 0.001432 & 0.04533 \\
RGG-7 & Nash & Combined & 4 & 100 & 95.0 & 0.0001131 & 0.004388 \\
RGG-7 & Nash & Combined & 8 & 100 & 100.0 & 0 & 0 \\
RGG-7 & Mixed SSE & Support & -- & 300 & 44.0 & 0.02182 & 0.1128 \\
RGG-7 & Mixed SSE & Surrogate & 1 & 100 & 90.0 & 0.001153 & 0.03095 \\
RGG-7 & Mixed SSE & Surrogate & 4 & 100 & 96.0 & 0.0002893 & 0.01784 \\
RGG-7 & Mixed SSE & Surrogate & 8 & 100 & 100.0 & 0 & 0 \\
RGG-7 & Mixed SSE & Combined & 1 & 100 & 94.0 & 0.0006107 & 0.03095 \\
RGG-7 & Mixed SSE & Combined & 4 & 100 & 98.0 & 0.0001869 & 0.01784 \\
RGG-7 & Mixed SSE & Combined & 8 & 100 & 100.0 & 0 & 0 \\
RGG-8 & Nash & Support & -- & 300 & 55.0 & 0.02226 & 0.2543 \\
RGG-8 & Nash & Surrogate & 1 & 100 & 81.0 & 0.00441 & 0.07968 \\
RGG-8 & Nash & Surrogate & 4 & 100 & 87.0 & 0.002687 & 0.05024 \\
RGG-8 & Nash & Surrogate & 8 & 100 & 96.0 & 0.001016 & 0.04166 \\
RGG-8 & Nash & Combined & 1 & 100 & 85.0 & 0.002296 & 0.07968 \\
RGG-8 & Nash & Combined & 4 & 100 & 89.0 & 0.001096 & 0.03307 \\
RGG-8 & Nash & Combined & 8 & 100 & 96.0 & 0.0002627 & 0.01262 \\
RGG-8 & Mixed SSE & Support & -- & 300 & 28.0 & 0.02594 & 0.213 \\
RGG-8 & Mixed SSE & Surrogate & 1 & 100 & 86.0 & 0.004725 & 0.2162 \\
RGG-8 & Mixed SSE & Surrogate & 4 & 100 & 89.0 & 0.002316 & 0.07554 \\
RGG-8 & Mixed SSE & Surrogate & 8 & 100 & 91.0 & 0.001666 & 0.07554 \\
RGG-8 & Mixed SSE & Combined & 1 & 100 & 92.0 & 0.001098 & 0.04319 \\
RGG-8 & Mixed SSE & Combined & 4 & 100 & 93.0 & 0.001021 & 0.04319 \\
RGG-8 & Mixed SSE & Combined & 8 & 100 & 96.0 & 0.0003409 & 0.01761 \\
\end{longtable}
\end{landscape}

\begin{landscape}
\scriptsize
\setlength{\tabcolsep}{3pt}
\begin{longtable}{@{}lllrrrrrr@{}}
\caption{Optimization over pure configurations by topology, tie-breaking rule, and method.
Gap columns are percentages of the enumerated defender payoff range; runtime is in
milliseconds.}\label{tab:detail-radar}\\
\toprule
Topology & Tie rule & Method & Feasible (\%) & Exact (\%) & Mean gap (\%) & Max gap (\%) & Time (ms) & Evals \\
\midrule
\endfirsthead
\multicolumn{9}{l}{\tablename~\thetable\ continued from the previous page}\\
\toprule
Topology & Tie rule & Method & Feasible (\%) & Exact (\%) & Mean gap (\%) & Max gap (\%) & Time (ms) & Evals \\
\midrule
\endhead
\midrule
\multicolumn{9}{r}{Continued on the next page}\\
\endfoot
\bottomrule
\endlastfoot
Diamond-4 & Pessimistic & Exact & 100.0 & 100.0 & 0 & 0 & 0.04092 & 33.57 \\
Diamond-4 & Pessimistic & \texttt{RADAR} & 100.0 & 82.1 & 0.2207 & 3.338 & 0.05098 & 41.79 \\
Diamond-4 & Pessimistic & RAD & 100.0 & 67.0 & 0.4977 & 5.421 & 0.03597 & 32 \\
Diamond-4 & Pessimistic & Greedy & 100.0 & 0.0 & 35.81 & 61.5 & 0.00658 & 5 \\
Diamond-4 & Strong & Exact & 100.0 & 100.0 & 0 & 0 & 0.04654 & 33.57 \\
Diamond-4 & Strong & \texttt{RADAR} & 100.0 & 82.1 & 0.2207 & 3.338 & 0.05523 & 41.79 \\
Diamond-4 & Strong & RAD & 100.0 & 67.0 & 0.4977 & 5.421 & 0.03495 & 32 \\
Diamond-4 & Strong & Greedy & 100.0 & 0.0 & 35.81 & 61.5 & 0.00637 & 5 \\
Diamond-7 & Pessimistic & Exact & 100.0 & 100.0 & 0 & 0 & 4.277 & 3657 \\
Diamond-7 & Pessimistic & \texttt{RADAR} & 100.0 & 22.7 & 2.023 & 9.673 & 0.2398 & 182.3 \\
Diamond-7 & Pessimistic & RAD & 100.0 & 3.7 & 4.911 & 10.66 & 0.1314 & 116 \\
Diamond-7 & Pessimistic & Greedy & 100.0 & 0.0 & 74.04 & 77.99 & 0.01285 & 9 \\
Diamond-7 & Strong & Exact & 100.0 & 100.0 & 0 & 0 & 4.399 & 3657 \\
Diamond-7 & Strong & \texttt{RADAR} & 100.0 & 22.7 & 2.023 & 9.673 & 0.2548 & 182.3 \\
Diamond-7 & Strong & RAD & 100.0 & 3.7 & 4.911 & 10.66 & 0.1313 & 116 \\
Diamond-7 & Strong & Greedy & 100.0 & 0.0 & 74.04 & 77.99 & 0.01346 & 9 \\
RGG-6 & Pessimistic & Exact & 100.0 & 100.0 & 0 & 0 & 2.484 & 2212 \\
RGG-6 & Pessimistic & \texttt{RADAR} & 100.0 & 55.1 & 7.112 & 40.69 & 0.2845 & 226.1 \\
RGG-6 & Pessimistic & RAD & 100.0 & 23.0 & 12.48 & 62.27 & 0.1836 & 163.4 \\
RGG-6 & Pessimistic & Greedy & 100.0 & 23.0 & 42.91 & 76.77 & 0.02213 & 17.12 \\
RGG-6 & Strong & Exact & 100.0 & 100.0 & 0 & 0 & 2.492 & 2212 \\
RGG-6 & Strong & \texttt{RADAR} & 100.0 & 55.1 & 7.112 & 40.69 & 0.2963 & 226.1 \\
RGG-6 & Strong & RAD & 100.0 & 23.0 & 12.48 & 62.27 & 0.1819 & 163.4 \\
RGG-6 & Strong & Greedy & 100.0 & 23.0 & 42.91 & 76.77 & 0.02264 & 17.12 \\
RGG-7 & Pessimistic & Exact & 100.0 & 100.0 & 0 & 0 & 19.7 & 16{,}850 \\
RGG-7 & Pessimistic & \texttt{RADAR} & 100.0 & 43.4 & 5.066 & 45.53 & 0.4789 & 348.7 \\
RGG-7 & Pessimistic & RAD & 100.0 & 10.6 & 14.08 & 49.04 & 0.2729 & 236.6 \\
RGG-7 & Pessimistic & Greedy & 100.0 & 27.7 & 39.22 & 77.96 & 0.03716 & 26.03 \\
RGG-7 & Strong & Exact & 100.0 & 100.0 & 0 & 0 & 19.69 & 16{,}850 \\
RGG-7 & Strong & \texttt{RADAR} & 100.0 & 43.4 & 5.066 & 45.53 & 0.4942 & 348.7 \\
RGG-7 & Strong & RAD & 100.0 & 10.6 & 14.08 & 49.04 & 0.2747 & 236.6 \\
RGG-7 & Strong & Greedy & 100.0 & 27.7 & 39.22 & 77.96 & 0.03668 & 26.03 \\
RGG-8 & Pessimistic & Exact & 100.0 & 100.0 & 0 & 0 & 68.6 & 56{,}810 \\
RGG-8 & Pessimistic & \texttt{RADAR} & 100.0 & 33.1 & 3.777 & 48.36 & 0.5889 & 412.1 \\
RGG-8 & Pessimistic & RAD & 100.0 & 11.3 & 11.6 & 75.47 & 0.3231 & 275.9 \\
RGG-8 & Pessimistic & Greedy & 100.0 & 21.4 & 22.28 & 77.64 & 0.05399 & 38.32 \\
RGG-8 & Strong & Exact & 100.0 & 100.0 & 0 & 0 & 68.63 & 56{,}810 \\
RGG-8 & Strong & \texttt{RADAR} & 100.0 & 33.1 & 3.777 & 48.36 & 0.5976 & 412.1 \\
RGG-8 & Strong & RAD & 100.0 & 11.3 & 11.6 & 75.47 & 0.3233 & 275.9 \\
RGG-8 & Strong & Greedy & 100.0 & 21.4 & 22.28 & 77.64 & 0.05393 & 38.32 \\
\end{longtable}
\end{landscape}

\begin{landscape}
\scriptsize
\setlength{\tabcolsep}{3pt}
\begin{longtable}{@{}lllrrrrrrrr@{}}
\caption{\texttt{RGR} on enumerable games. Defender differences are percentages of the
enumerated defender payoff range.}\label{tab:detail-rgr-enumerable}\\
\toprule
Topology & Game & Candidate source & $K_{\mathrm{cand}}$ & \shortstack{Completed\\(\%)} & \shortstack{Mean\\$100\delta_D$} & \shortstack{Mean\\$100|\delta_D|$} & \shortstack{Max\\$100|\delta_D|$} & \shortstack{Mean route\\gap} & \shortstack{Mean generated\\action gap} & Runs \\
\midrule
\endfirsthead
\multicolumn{11}{l}{\tablename~\thetable\ continued from the previous page}\\
\toprule
Topology & Game & Candidate source & $K_{\mathrm{cand}}$ & \shortstack{Completed\\(\%)} & \shortstack{Mean\\$100\delta_D$} & \shortstack{Mean\\$100|\delta_D|$} & \shortstack{Max\\$100|\delta_D|$} & \shortstack{Mean route\\gap} & \shortstack{Mean generated\\action gap} & Runs \\
\midrule
\endhead
\midrule
\multicolumn{11}{r}{Continued on the next page}\\
\endfoot
\bottomrule
\endlastfoot
Diamond-4 & Nash & Support & -- & 100.0 & 2.058 & 2.058 & 4.34 & 0 & 0 & 300 \\
Diamond-4 & Nash & Surrogate & 1 & 100.0 & 2.058 & 2.058 & 4.34 & 0 & 0 & 100 \\
Diamond-4 & Nash & Surrogate & 4 & 100.0 & 2.058 & 2.058 & 4.34 & 0 & 0 & 100 \\
Diamond-4 & Nash & Surrogate & 8 & 100.0 & 2.058 & 2.058 & 4.34 & 0 & 0 & 100 \\
Diamond-4 & Nash & Combined & 1 & 100.0 & 2.058 & 2.058 & 4.34 & 0 & 0 & 100 \\
Diamond-4 & Nash & Combined & 4 & 100.0 & 2.058 & 2.058 & 4.34 & 0 & 0 & 100 \\
Diamond-4 & Nash & Combined & 8 & 100.0 & 2.058 & 2.058 & 4.34 & 0 & 0 & 100 \\
Diamond-4 & Mixed SSE & Support & -- & 100.0 & 2.058 & 2.058 & 4.34 & 0 & 0 & 300 \\
Diamond-4 & Mixed SSE & Surrogate & 1 & 100.0 & 2.058 & 2.058 & 4.34 & 0 & 0 & 100 \\
Diamond-4 & Mixed SSE & Surrogate & 4 & 100.0 & 2.058 & 2.058 & 4.34 & 0 & 0 & 100 \\
Diamond-4 & Mixed SSE & Surrogate & 8 & 100.0 & 2.058 & 2.058 & 4.34 & 0 & 0 & 100 \\
Diamond-4 & Mixed SSE & Combined & 1 & 100.0 & 2.058 & 2.058 & 4.34 & 0 & 0 & 100 \\
Diamond-4 & Mixed SSE & Combined & 4 & 100.0 & 2.058 & 2.058 & 4.34 & 0 & 0 & 100 \\
Diamond-4 & Mixed SSE & Combined & 8 & 100.0 & 2.058 & 2.058 & 4.34 & 0 & 0 & 100 \\
Diamond-7 & Nash & Support & -- & 100.0 & 0 & 0 & 0 & 0.01196 & 0 & 300 \\
Diamond-7 & Nash & Surrogate & 1 & 100.0 & 0 & 0 & 0 & 0.0001277 & 0 & 100 \\
Diamond-7 & Nash & Surrogate & 4 & 100.0 & 0 & 0 & 0 & 0 & 0 & 100 \\
Diamond-7 & Nash & Surrogate & 8 & 100.0 & 0 & 0 & 0 & 0 & 0 & 100 \\
Diamond-7 & Nash & Combined & 1 & 100.0 & 0 & 0 & 0 & 0 & 0 & 100 \\
Diamond-7 & Nash & Combined & 4 & 100.0 & 0 & 0 & 0 & 0 & 0 & 100 \\
Diamond-7 & Nash & Combined & 8 & 100.0 & 0 & 0 & 0 & 0 & 0 & 100 \\
Diamond-7 & Mixed SSE & Support & -- & 100.0 & 0 & 0 & 0 & 0.009809 & 0 & 300 \\
Diamond-7 & Mixed SSE & Surrogate & 1 & 100.0 & 0 & 0 & 0 & 0.0003488 & 0 & 100 \\
Diamond-7 & Mixed SSE & Surrogate & 4 & 100.0 & 0 & 0 & 0 & 0 & 0 & 100 \\
Diamond-7 & Mixed SSE & Surrogate & 8 & 100.0 & 0 & 0 & 0 & 0 & 0 & 100 \\
Diamond-7 & Mixed SSE & Combined & 1 & 100.0 & 0 & 0 & 0 & 0 & 0 & 100 \\
Diamond-7 & Mixed SSE & Combined & 4 & 100.0 & 0 & 0 & 0 & 0 & 0 & 100 \\
Diamond-7 & Mixed SSE & Combined & 8 & 100.0 & 0 & 0 & 0 & 0 & 0 & 100 \\
RGG-6 & Nash & Support & -- & 100.0 & -0.0939 & 0.1683 & 1.569 & 0.01792 & 0 & 300 \\
RGG-6 & Nash & Surrogate & 1 & 100.0 & 0.054 & 0.0586 & 0.9174 & 0.0008331 & 0 & 100 \\
RGG-6 & Nash & Surrogate & 4 & 100.0 & 0.0478 & 0.0478 & 0.9174 & 0.0001421 & 0 & 100 \\
RGG-6 & Nash & Surrogate & 8 & 100.0 & 0.0402 & 0.0402 & 0.8186 & 0 & 0 & 100 \\
RGG-6 & Nash & Combined & 1 & 100.0 & 0.0421 & 0.0447 & 0.9174 & 0.000768 & 0 & 100 \\
RGG-6 & Nash & Combined & 4 & 100.0 & 0.0382 & 0.0382 & 0.9174 & 0 & 0 & 100 \\
RGG-6 & Nash & Combined & 8 & 100.0 & 0.04 & 0.04 & 0.8186 & 0 & 0 & 100 \\
RGG-6 & Mixed SSE & Support & -- & 100.0 & -0.1133 & 0.1658 & 1.569 & 0.01724 & 0 & 300 \\
RGG-6 & Mixed SSE & Surrogate & 1 & 100.0 & 0.0311 & 0.0357 & 0.696 & 0.0008019 & 0 & 100 \\
RGG-6 & Mixed SSE & Surrogate & 4 & 100.0 & 0.0441 & 0.0441 & 0.8186 & 0 & 0 & 100 \\
RGG-6 & Mixed SSE & Surrogate & 8 & 100.0 & 0.034 & 0.034 & 0.8186 & 0 & 0 & 100 \\
RGG-6 & Mixed SSE & Combined & 1 & 100.0 & 0.0323 & 0.0349 & 0.696 & 0.0005901 & 0 & 100 \\
RGG-6 & Mixed SSE & Combined & 4 & 100.0 & 0.0376 & 0.0376 & 0.8186 & 0 & 0 & 100 \\
RGG-6 & Mixed SSE & Combined & 8 & 100.0 & 0.0406 & 0.0406 & 0.8186 & 0 & 0 & 100 \\
RGG-7 & Nash & Support & -- & 100.0 & -0.0713 & 0.3994 & 15.14 & 0.02436 & 0 & 300 \\
RGG-7 & Nash & Surrogate & 1 & 100.0 & -0.1005 & 0.2133 & 15.69 & 0.00144 & 0 & 100 \\
RGG-7 & Nash & Surrogate & 4 & 100.0 & -0.0975 & 0.2162 & 15.69 & 0.000193 & 0 & 100 \\
RGG-7 & Nash & Surrogate & 8 & 100.0 & -0.0836 & 0.2301 & 15.69 & $2.756\times10^{-6}$ & 0 & 100 \\
RGG-7 & Nash & Combined & 1 & 100.0 & -0.0982 & 0.2155 & 15.69 & 0.001432 & 0 & 100 \\
RGG-7 & Nash & Combined & 4 & 100.0 & -0.0972 & 0.2165 & 15.69 & 0.0001131 & 0 & 100 \\
RGG-7 & Nash & Combined & 8 & 100.0 & -0.083 & 0.2308 & 15.69 & 0 & 0 & 100 \\
RGG-7 & Mixed SSE & Support & -- & 100.0 & -0.0505 & 0.1101 & 1.44 & 0.02182 & 0 & 300 \\
RGG-7 & Mixed SSE & Surrogate & 1 & 100.0 & 0.0577 & 0.0579 & 0.6459 & 0.001153 & 0 & 100 \\
RGG-7 & Mixed SSE & Surrogate & 4 & 100.0 & 0.0672 & 0.0672 & 1.231 & 0.0002893 & 0 & 100 \\
RGG-7 & Mixed SSE & Surrogate & 8 & 100.0 & 0.0633 & 0.0633 & 0.6459 & 0 & 0 & 100 \\
RGG-7 & Mixed SSE & Combined & 1 & 100.0 & 0.0593 & 0.0593 & 0.6459 & 0.0006107 & 0 & 100 \\
RGG-7 & Mixed SSE & Combined & 4 & 100.0 & 0.0621 & 0.0621 & 0.6459 & 0.0001869 & 0 & 100 \\
RGG-7 & Mixed SSE & Combined & 8 & 100.0 & 0.0644 & 0.0644 & 0.6459 & 0 & 0 & 100 \\
RGG-8 & Nash & Support & -- & 100.0 & 0 & 0 & 0 & 0.02226 & 0 & 300 \\
RGG-8 & Nash & Surrogate & 1 & 100.0 & 0.0062 & 0.0062 & 0.6211 & 0.00441 & 0 & 100 \\
RGG-8 & Nash & Surrogate & 4 & 100.0 & 0 & 0 & 0 & 0.002687 & 0 & 100 \\
RGG-8 & Nash & Surrogate & 8 & 100.0 & 0 & 0 & 0 & 0.001016 & 0 & 100 \\
RGG-8 & Nash & Combined & 1 & 100.0 & 0.0062 & 0.0062 & 0.6211 & 0.002296 & 0 & 100 \\
RGG-8 & Nash & Combined & 4 & 100.0 & 0 & 0 & 0 & 0.001096 & 0 & 100 \\
RGG-8 & Nash & Combined & 8 & 100.0 & 0 & 0 & 0 & 0.0002627 & 0 & 100 \\
RGG-8 & Mixed SSE & Support & -- & 100.0 & 0 & 0 & 0 & 0.02594 & 0 & 300 \\
RGG-8 & Mixed SSE & Surrogate & 1 & 100.0 & 0 & 0 & 0 & 0.004725 & 0 & 100 \\
RGG-8 & Mixed SSE & Surrogate & 4 & 100.0 & 0 & 0 & 0 & 0.002316 & 0 & 100 \\
RGG-8 & Mixed SSE & Surrogate & 8 & 100.0 & 0 & 0 & 0 & 0.001666 & 0 & 100 \\
RGG-8 & Mixed SSE & Combined & 1 & 100.0 & 0 & 0 & 0 & 0.001098 & 0 & 100 \\
RGG-8 & Mixed SSE & Combined & 4 & 100.0 & 0 & 0 & 0 & 0.001021 & 0 & 100 \\
RGG-8 & Mixed SSE & Combined & 8 & 100.0 & 0 & 0 & 0 & 0.0003409 & 0 & 100 \\
\end{longtable}
\end{landscape}

\begin{landscape}
\scriptsize
\setlength{\tabcolsep}{3pt}
\begin{longtable}{@{}lllrrrr@{}}
\caption{\texttt{RGR} on larger games without exhaustive defender payoff comparisons.}\label{tab:detail-rgr-large}\\
\toprule
Topology & Game & Candidate source & $K_{\mathrm{cand}}$ & Runs & \shortstack{Completed\\(\%)} & \shortstack{Mean generated\\action gap} \\
\midrule
\endfirsthead
\multicolumn{7}{l}{\tablename~\thetable\ continued from the previous page}\\
\toprule
Topology & Game & Candidate source & $K_{\mathrm{cand}}$ & Runs & \shortstack{Completed\\(\%)} & \shortstack{Mean generated\\action gap} \\
\midrule
\endhead
\midrule
\multicolumn{7}{r}{Continued on the next page}\\
\endfoot
\bottomrule
\endlastfoot
Diamond-10 & Nash & Support & -- & 300 & 100.0 & 0 \\
Diamond-10 & Nash & Surrogate & 1 & 100 & 100.0 & 0 \\
Diamond-10 & Nash & Surrogate & 4 & 100 & 100.0 & 0 \\
Diamond-10 & Nash & Surrogate & 8 & 100 & 100.0 & 0 \\
Diamond-10 & Nash & Combined & 1 & 100 & 100.0 & 0 \\
Diamond-10 & Nash & Combined & 4 & 100 & 100.0 & 0 \\
Diamond-10 & Nash & Combined & 8 & 100 & 100.0 & 0 \\
Diamond-10 & Mixed SSE & Support & -- & 300 & 100.0 & 0 \\
Diamond-10 & Mixed SSE & Surrogate & 1 & 100 & 100.0 & 0 \\
Diamond-10 & Mixed SSE & Surrogate & 4 & 100 & 100.0 & 0 \\
Diamond-10 & Mixed SSE & Surrogate & 8 & 100 & 100.0 & 0 \\
Diamond-10 & Mixed SSE & Combined & 1 & 100 & 100.0 & 0 \\
Diamond-10 & Mixed SSE & Combined & 4 & 100 & 100.0 & 0 \\
Diamond-10 & Mixed SSE & Combined & 8 & 100 & 100.0 & 0 \\
RGG-9 & Nash & Support & -- & 300 & 100.0 & 0 \\
RGG-9 & Nash & Surrogate & 1 & 100 & 100.0 & 0 \\
RGG-9 & Nash & Surrogate & 4 & 100 & 100.0 & 0 \\
RGG-9 & Nash & Surrogate & 8 & 100 & 100.0 & 0 \\
RGG-9 & Nash & Combined & 1 & 100 & 100.0 & 0 \\
RGG-9 & Nash & Combined & 4 & 100 & 100.0 & 0 \\
RGG-9 & Nash & Combined & 8 & 100 & 100.0 & 0.001059 \\
RGG-9 & Mixed SSE & Support & -- & 300 & 100.0 & 0 \\
RGG-9 & Mixed SSE & Surrogate & 1 & 100 & 100.0 & 0.0006315 \\
RGG-9 & Mixed SSE & Surrogate & 4 & 100 & 100.0 & 0 \\
RGG-9 & Mixed SSE & Surrogate & 8 & 100 & 100.0 & 0 \\
RGG-9 & Mixed SSE & Combined & 1 & 100 & 100.0 & 0 \\
RGG-9 & Mixed SSE & Combined & 4 & 100 & 100.0 & 0.0007769 \\
RGG-9 & Mixed SSE & Combined & 8 & 100 & 100.0 & 0.0008968 \\
RGG-10 & Nash & Support & -- & 300 & 91.0 & 0.1803 \\
RGG-10 & Nash & Surrogate & 1 & 100 & 88.0 & 0.2955 \\
RGG-10 & Nash & Surrogate & 4 & 100 & 84.0 & 0.2212 \\
RGG-10 & Nash & Surrogate & 8 & 100 & 87.0 & 0.252 \\
RGG-10 & Nash & Combined & 1 & 100 & 88.0 & 0.2029 \\
RGG-10 & Nash & Combined & 4 & 100 & 86.0 & 0.2273 \\
RGG-10 & Nash & Combined & 8 & 100 & 87.0 & 0.2472 \\
RGG-10 & Mixed SSE & Support & -- & 300 & 97.0 & 0.02844 \\
RGG-10 & Mixed SSE & Surrogate & 1 & 100 & 100.0 & 0 \\
RGG-10 & Mixed SSE & Surrogate & 4 & 100 & 100.0 & 0 \\
RGG-10 & Mixed SSE & Surrogate & 8 & 100 & 100.0 & $1.789\times10^{-5}$ \\
RGG-10 & Mixed SSE & Combined & 1 & 100 & 99.0 & 0.004716 \\
RGG-10 & Mixed SSE & Combined & 4 & 100 & 99.0 & 0.004699 \\
RGG-10 & Mixed SSE & Combined & 8 & 100 & 100.0 & $9.791\times10^{-5}$ \\
\end{longtable}
\end{landscape}

\begin{landscape}
\scriptsize
\setlength{\tabcolsep}{1.5pt}
\begin{longtable}{@{}
>{\raggedright\arraybackslash}p{0.10\linewidth}
>{\raggedright\arraybackslash}p{0.29\linewidth}
>{\raggedright\arraybackslash}p{0.29\linewidth}
>{\centering\arraybackslash}p{0.07\linewidth}
>{\centering\arraybackslash}p{0.07\linewidth}
>{\centering\arraybackslash}p{0.07\linewidth}
>{\centering\arraybackslash}p{0.07\linewidth}@{}}
\caption{Complete numerical checks of the predicted Nash--Stackelberg payoff relationships.
Mean and maximum entries are residuals, and ``Bound'' is the allowed residual.}\label{tab:detail-relations}\\
\toprule
Profile & Evaluator check & Prediction & Passed (\%) & Mean & Max & Bound \\
\midrule
\endfirsthead
\multicolumn{7}{l}{\tablename~\thetable\ continued from the previous page}\\
\toprule
Profile & Evaluator check & Prediction & Passed (\%) & Mean & Max & Bound \\
\midrule
\endhead
\midrule
\multicolumn{7}{r}{Continued on the next page}\\
\endfoot
\bottomrule
\endlastfoot
Common ratio & \nolinkurl{common-ratio-attacker-gap-bounds} & $0\le U_A(\text{mixed SSE})-U_A(\text{mixed NE})\le$ action constant range$/\lambda$ & 100.0 & 0 & 0 & 0 \\
Common ratio & \nolinkurl{common-ratio-defender-gap-bounds} & $0\le U_D(\text{mixed SSE})-U_D(\text{mixed NE})\le$ action constant range & 100.0 & 0 & 0 & 0 \\
Common ratio & \nolinkurl{common-ratio-NE-minimax-value} & $U_A(\text{mixed NE})$ equals the attacker minimax value & 100.0 & 0 & 0 & 0 \\
Common ratio & \nolinkurl{mixed-SSE-defender-dominates-NE} & $U_D(\text{mixed SSE})\ge U_D(\text{mixed NE})$ & 100.0 & 0 & 0 & 0 \\
Common ratio & \nolinkurl{pure-SSE-is-NE-defender-deviation} & Pure SSE outcome has no profitable defender deviation & 0.0 & 8 & 8 & 0 \\
Common ratio & \nolinkurl{strategic-zero-sum-identity} & $U_D+\lambda U_A=$ action constant & 100.0 & 0 & 0 & 0 \\
Heterogeneous & \nolinkurl{mixed-SSE-defender-dominates-NE} & $U_D(\text{mixed SSE})\ge U_D(\text{mixed NE})$ & 100.0 & 0 & 0 & 0 \\
Heterogeneous & \nolinkurl{pure-SSE-is-NE-defender-deviation} & Pure SSE outcome has no profitable defender deviation & 0.0 & 5 & 5 & 0 \\
Heterogeneous & \nolinkurl{single-route-mixed-SSE-attacker-bound} & $U_A(\text{mixed SSE})\ge U_A(\text{mixed NE})$ & 100.0 & 0 & 0 & 0 \\
Heterogeneous & \nolinkurl{single-route-NE-minimax-value} & $U_A(\text{mixed NE})$ equals the attacker minimax value & 100.0 & 0 & 0 & 0 \\
Heterogeneous & \nolinkurl{single-route-pure-SSE-attacker-bound} & $U_A(\text{pure SSE})\ge U_A(\text{mixed NE})$ & 100.0 & 0 & 0 & 0 \\
Weighted zero sum & \nolinkurl{mixed-SSE-defender-dominates-NE} & $U_D(\text{mixed SSE})\ge U_D(\text{mixed NE})$ & 100.0 & 0 & 0 & 0 \\
Weighted zero sum & \nolinkurl{pure-SSE-is-NE-defender-deviation} & Pure SSE outcome has no profitable defender deviation & 0.0 & 8 & 8 & 0 \\
Weighted zero sum & \nolinkurl{randomization-gain-identity} & $\Delta U_D=\lambda[U_A(\text{pure SSE})-U_A(\text{mixed SSE})]$ & 100.0 & 0 & 0 & 0 \\
Weighted zero sum & \nolinkurl{weighted-zero-sum-NE-minimax-value} & $U_A(\text{mixed NE})$ equals the attacker minimax value & 100.0 & 0 & 0 & 0 \\
Weighted zero sum & \nolinkurl{weighted-zero-sum-payoff-identity} & $U_D+\lambda U_A=0$ & 100.0 & 0 & 0 & 0 \\
Weighted zero sum & \nolinkurl{weighted-zero-sum-pure-mixed-attacker-bound} & $U_A(\text{mixed SSE})\le U_A(\text{pure SSE})$ & 100.0 & 0 & 0 & 0 \\
Weighted zero sum & \nolinkurl{weighted-zero-sum-pure-mixed-defender-bound} & $U_D(\text{pure SSE})\le U_D(\text{mixed SSE})$ & 100.0 & 0 & 0 & 0 \\
Weighted zero sum & \nolinkurl{weighted-zero-sum-SSE-NE-attacker-equality} & $U_A(\text{mixed SSE})=U_A(\text{mixed NE})$ & 100.0 & 0 & 0 & 0 \\
Weighted zero sum & \nolinkurl{weighted-zero-sum-SSE-NE-defender-equality} & $U_D(\text{mixed SSE})=U_D(\text{mixed NE})$ & 100.0 & 0 & 0 & 0 \\
\end{longtable}
\end{landscape}

\clearpage
\bibliographystyle{IEEEtran}
\bibliography{references}
\end{document}